\documentclass[a4paper,UKenglish,cleveref,autoref]{lipics-v2021}

\hideLIPIcs  
\nolinenumbers

\usepackage[utf8]{inputenc}

\usepackage[all,british]{foreign}

\usepackage{amsmath}
\usepackage{amssymb}
\usepackage{mathtools}
\usepackage{bbold} 
\usepackage{fancyvrb}
\usepackage{tcolorbox}
\usepackage{mathpartir}
\usepackage{xspace}
\usepackage[disable]{todonotes}
\usepackage{cancel}
\usepackage{stackengine} 

\usepackage{varwidth}
\newcommand\codeblock[1]{%
  {\begin{varwidth}{0.9\textwidth}#1\end{varwidth}}
} 
\usepackage{agda}

\usepackage{newunicodechar}

\newcommand{\AD}[1]{\AgdaDatatype{#1}}
\newcommand{\AC}[1]{\AgdaInductiveConstructor{#1}}
\newcommand{\AF}[1]{\AgdaFunction{#1}}
\newcommand{\AB}[1]{\AgdaBound{#1}}
\newcommand{\AK}[1]{\AgdaKeyword{#1}}

\newcommand{\AM}[1]{\AgdaMacro{#1}}

\newunicodechar{∀}{\ensuremath{\forall}}
\newunicodechar{∃}{\ensuremath{\exists}}
\newunicodechar{≡}{\ensuremath{\equiv}}
\newunicodechar{≈}{\ensuremath{\approx}}
\newunicodechar{≟}{\ensuremath{\stackrel{?}{=}}}
\newunicodechar{′}{{'}}
\newunicodechar{∷}{\ensuremath{\dblcolon}}
\newunicodechar{⊤}{\ensuremath{\top}}

\newunicodechar{₀}{\ensuremath{_0}}
\newunicodechar{₁}{\ensuremath{_1}}
\newunicodechar{₂}{\ensuremath{_2}}
\newunicodechar{₃}{\ensuremath{_3}}
\newunicodechar{₄}{\ensuremath{_4}}
\newunicodechar{₅}{\ensuremath{_5}}
\newunicodechar{₆}{\ensuremath{_6}}
\newunicodechar{₇}{\ensuremath{_7}}
\newunicodechar{₈}{\ensuremath{_8}}
\newunicodechar{₉}{\ensuremath{_9}}
\newunicodechar{ι}{\ensuremath{\iota}}
\newunicodechar{ₗ}{\ensuremath{_l}}
\newunicodechar{ₐ}{\ensuremath{_a}}
\newunicodechar{ₖ}{\ensuremath{_k}}
\newunicodechar{ᵢ}{\ensuremath{_i}}
\newunicodechar{ⱼ}{\ensuremath{_j}}
\newunicodechar{ᵥ}{\ensuremath{_v}}
\newunicodechar{ₕ}{\ensuremath{_h}}
\newunicodechar{ₚ}{\ensuremath{_p}}
\newunicodechar{ₛ}{\ensuremath{_s}}
\newunicodechar{ₒ}{\ensuremath{_o}}
\newunicodechar{ₙ}{\ensuremath{_n}}
\newunicodechar{ᵣ}{\ensuremath{_r}}

\newunicodechar{ⁿ}{\ensuremath{^n}}
\newunicodechar{ᶠ}{\ensuremath{^f}}
\newunicodechar{ᶜ}{\ensuremath{^c}}
\newunicodechar{ʳ}{\ensuremath{^r}}
\newunicodechar{ˡ}{\ensuremath{^l}}
\newunicodechar{ᵛ}{\ensuremath{^v}}
\newunicodechar{ᵉ}{\ensuremath{^e}}
\newunicodechar{ℕ}{\ensuremath{\mathbb{N}}}
\newunicodechar{ℝ}{\ensuremath{\mathbb{R}}}
\newunicodechar{ℂ}{\ensuremath{\mathbb{C}}}
\newunicodechar{𝔽}{\ensuremath{\mathbb{F}}}
\newunicodechar{𝟘}{\ensuremath{\mathbb{0}}}
\newunicodechar{𝟙}{\ensuremath{\mathbb{1}}}

\newunicodechar{π}{\ensuremath{\pi}}
\newunicodechar{ω}{\ensuremath{\omega}}

\title{When Agda met Vampire}

\author{Artjoms \v{S}inkarovs}{Electronics and Computer Science, University of Southampton, UK}{a.sinkarovs@soton.ac.uk}{https://orcid.org/0000-0003-3292-2985}{[funding]} 
\author{Michael Rawson}{Electronics and Computer Science, University of Southampton, UK}{michael@rawsons.uk}{https://orcid.org/0000-0001-7834-1567}{[funding]}

\authorrunning{A. \v{S}inkarovs and M. Rawson} 

\Copyright{Artjoms \v{S}inkarovs and Michael Rawson} 



\category{} 

\relatedversion{} 

\acknowledgements{We would like to thank Mario Carneiro for introducing us to the proof-transform technique used here.}

\EventEditors{John Q. Open and Joan R. Access}
\EventNoEds{2}
\EventLongTitle{42nd Conference on Very Important Topics (CVIT 2016)}
\EventShortTitle{CVIT 2016}
\EventAcronym{CVIT}
\EventYear{2016}
\EventDate{December 24--27, 2016}
\EventLocation{Little Whinging, United Kingdom}
\EventLogo{}
\SeriesVolume{42}
\ArticleNo{23}

\begin{document}
\newcommand{\vampire}{\textsc{Vampire}\xspace}
\newcommand{\mgu}{\mathrm{mgu}}
\newcommand{\pa}{\vdash_{\mathrm{PA}}}
\newcommand{\ha}{\vdash_{\mathrm{HA}}}

\maketitle

\begin{abstract}
  Dependently-typed proof assistants furnish expressive foundations for
  mechanised mathematics and verified software.
  However, automation for these systems has been either modest in scope or complex in implementation.
  We aim to improve the situation by integrating proof assistants with automated theorem
  provers (ATPs) in a simple way, while preserving the correctness guarantees of the former.
  A central difficulty arises from the fact that most ATPs operate in classical
  first‑order logic, whereas these proof assistants are grounded in constructive
  dependent type theory. We identify an expressive fragment of both
  languages---essentially equational Horn---that admits sound, straightforward translations in both directions.
  The approach produces a prototype system for Agda forwarding
  proof obligations to the ATP \vampire, then transforming the resulting classical
  proof into a constructive proof term that Agda can type‑check. The
  prototype automatically derives proofs concerning the properties of a complex field equipped with roots of unity,
  which took professional Agda developers two full days to complete.
  The required engineering effort is modest, and we anticipate that the
  methodology will extend readily to other ATPs and proof assistants.
\end{abstract}

\section{\label{sec:intro}Introduction}

Dependently typed proof assistants such as Lean, Agda and Rocq are powerful
environments for correct‑by‑construction programming and formal theorem
proving. Thanks to the Curry–Howard correspondence, program invariants or
theorems are expressed as types, and correctness follows once the corresponding
specification type‑checks.
This approach does not depend on untrusted components such as
SAT/SMT solvers, but it also places a substantial burden on the developer:
non‑trivial proof obligations frequently arise and must be discharged. In
practice many of these obligations are routine rather than mathematically deep
--- for example, proving that two constructions in a particular encoding are
isomorphic or share a simple property --- yet they are nevertheless
labour‑intensive. Consequently, effective proof automation is essential for the
practical adoption of dependently-typed systems.

Proof automation in interactive theorem provers (ITPs) is a longstanding and
difficult research problem. Most proof assistants provide tactic languages to
automate some work, but designing efficient proof‑search strategies for a given
theory remains hard. Na\"ive brute‑force search succeeds only in the simplest
cases; more capable methods lead naturally to Automatic Theorem Provers (ATPs).
At that juncture an ITP faces a choice: rely on external ATPs
--- which raises compatibility and trust issues ---
or implement proof search procedures inside the ITP.
Typical internal procedures are significantly weaker than the state of the art in ATP.
Practical ``hammer'' systems such as Isabelle's Sledgehammer or CoqHammer
use an external ATP to find an initial proof, then use that proof as a ``hint''
to reconstruct another proof using internal certified procedures.
The hint could be merely the premises required to prove a goal,
or using intermediate steps from the untrusted proof.
Sledgehammer works particularly well because the internal prover \emph{metis}
is relatively strong, and because Isabelle/HOL's logic is similar to that of ATPs.
Dependently-typed systems do not enjoy these advantages.

Here we pursue a deliberately lightweight variant of this approach,
playing to the strengths of existing systems.
We identify a useful fragment of a dependently-typed language that
can be translated straightforwardly to an ATP \emph{and vice-versa},
and produce a concise, verifiable round‑trip translation:
(i) translate a goal and the necessary lemmas into the ATP;
(ii) invoke the ATP to search for a proof; (iii) reconstruct the ATP proof as an
ITP proof term; (iv) typecheck the reconstructed term. By confining the
translation to a well‑chosen fragment, we avoid intrusive changes to either
system while still enabling substantial automation.
Proof re-construction also becomes substantially easier: if a proof is found by the external ATP,
it is almost always reconstructed.

Concretely, we integrate Agda --- based on Martin‑Löf type theory --- with
\vampire, a leading ATP and multi‑time winner of CADE ATP System Competition (CASC). We use
Agda's reflection facilities to extract the current goal and environment and to
emit Vampire problems. This reflection‑based technique does not require any
changes to Agda. \vampire's proof output, when successful, is translated back into
Agda using a compact reconstruction engine implemented in Prolog: Prolog's
native search facilities map naturally to the reconstruction task. The
resulting proof term is then certified by Agda's typechecker.

A central claim of this work is that such lightweight integration --- applying
each tool for what it already does best --- demands minimal implementation
effort and maintenance compared with deep, invasive integration. Our prototype
implementation is already practical: we used it to discharge
properties of a representation of the complex field with roots of unity that
arose in a real research project. What previously required an Agda specialist
around two full days is produced automatically.
The remainder of the paper describes the translated fragment, translation
and reconstruction, implementation details, and a case study.

\section{Background}
We assume passing familiarity with classical first-order logic and with dependently-typed interactive theorem proving.
There is a well-known correspondence~\cite{propositions-as-types} between logic and type systems, often under the slogan \emph{propositions as types}.
The gist is that logical propositions $P$ correspond to types $\tau$, and that proofs of $P$ correspond to terms inhabiting $\tau$.
This is the conceit of type-theoretic proof assistants like Agda: with a sufficiently expressive type system, one can prove non-trivial propositions by producing a term which inhabits a type.
Such systems exploit the propositions-as-types correspondence in one direction.
We will need to traverse both directions for the integration we seek.

One could be forgiven for thinking that the correspondence is tighter than it is: that there is no real difference between logic and type systems.
Unfortunately for us, there \emph{are} real technical and cultural differences, which are usually poorly understood by specialists in either world.
The best-known difference is a cultural preference for constructivism in type-theoretic systems, whereas logical proof asssistants and ATPs are more often classical.

Less well-understood in our experience is that logical systems --- including \vampire~--- draw a distinction between \emph{terms} and \emph{formulae}.
Terms are inhabitants of the \emph{domain of discourse}, which are implicitly assumed to be non-empty: this renders questions about whether a type is inhabited uninteresting, as all types are inhabited.
Deduction instead happens by deriving new formulae from old using the logic of the system.
Perhaps confusingly, terms can be Boolean~\cite{fool}.
On the other hand, type-theoretic systems often make no such distinction: there are simply terms, stratified by the has-type relation $t : \tau$.
Our task here is to reconcile the type system of Agda and the logic of \vampire: we will now describe both, after introducing some notation.

\subsection{Horn Logic and Notation}
We will make use of the Horn fragment of first-order logic, which we now briefly recall.
An \emph{atom} $P(t_1, \ldots, t_n)$ is a predicate $P$ applied to terms $t_i$, and a \emph{literal} $L$ is either an atom $A$ or its negation $\lnot A$.
Equations $s = t$ are merely a special atom with two arguments.
A first-order \emph{clause} is a universally-quantified disjunction of literals, from which we will usually elide the quantifiers.
A Horn clause is a clause with no more than one positive ``head'' literal.
We distinguish between \emph{definite} clauses with \emph{exactly} one positive literal, and \emph{goal} clauses with \emph{no} positive literals.
A definite clause can be written as an implication $P_1 \wedge \ldots \wedge P_n \to Q$, where each $P_i$ and $Q$ are atoms.
Naturally, a goal clause can then be written $P_1 \wedge \ldots \wedge P_n \to \bot$, although it is really the clause $\lnot P_1 \vee \ldots \vee \lnot P_n$.
The empty clause $\bot$ is also a goal clause, with $\bot$ for the head and an empty body.

When writing inference rules, we will use the letters $C$ and $D$ to indicate \emph{part} of a (Horn) clause, which may be empty.
Two clauses which participate in an inference are assumed to have their variables renamed apart from each other.
We will also assume that literals in a (Horn) clause may be re-ordered in order to apply an inference rule, as is common implicit practice.
Equally, equations $s = t$ are treated as implicitly symmetric.
We write $s[t]$ to indicate that $t$ is a subterm of $s$: this includes the case where $s$ is $t$.

We say that a first-order term (literal, clause) is \emph{ground} if it contains no variables.
A substitution $\sigma$ is a map from variables to terms: when applied to a term (literal, clause) it acts to replace every variable $x$ in that term with $\sigma(x)$.
Clearly, then, $\sigma(t) = t$ when $t$ is ground.
A substitution is a unifier for $s$ and $t$ when $\sigma(s) = \sigma(t)$, and it is the \emph{most general} unifier when it cannot be expressed as the composition of a unifier and a non-empty substitution.
Most general unifiers can be computed with a \emph{unification algorithm}.

\subsection{\label{sec:agda-intro}The Agda Proof Assistant}
\begin{code}[hide]%
\>[0]\AgdaComment{--\ module\ papers.ITP26.agda-intro\ where}\<%
\\
\>[0]\AgdaKeyword{module}\AgdaSpace{}%
\AgdaModule{\AgdaUnderscore{}}\AgdaSpace{}%
\AgdaKeyword{where}\<%
\\
\>[0]\AgdaKeyword{module}\AgdaSpace{}%
\AgdaModule{Intro}\AgdaSpace{}%
\AgdaKeyword{where}\<%
\\
\>[0][@{}l@{\AgdaIndent{0}}]%
\>[2]\AgdaKeyword{open}\AgdaSpace{}%
\AgdaKeyword{import}\AgdaSpace{}%
\AgdaModule{Data.Nat}\AgdaSpace{}%
\AgdaKeyword{hiding}\AgdaSpace{}%
\AgdaSymbol{(}\AgdaOperator{\AgdaFunction{\AgdaUnderscore{}<\AgdaUnderscore{}}}\AgdaSymbol{)}\<%
\\
\>[2]\AgdaKeyword{open}\AgdaSpace{}%
\AgdaKeyword{import}\AgdaSpace{}%
\AgdaModule{Data.Unit}\<%
\end{code}

Agda is both a dependently typed functional programming language and a
proof assistant, offering an interactive environment for writing and
checking formal proofs.  It provides inductive families, analogous to
Haskell's GADTs, but with indices that may be values as well as types.
The system is founded on intuitionistic type theory, a framework for
constructive mathematics developed by Per Martin‑L\"of.  It shares many
characteristics with other dependently typed proof assistants such as
Rocq, Idris and Lean.  The following is a concise overview of the
constructs necessary for the remainder of the paper; accessible
introductions are available online.\footnote{See \url{https://agda.readthedocs.io/en/latest/getting-started/tutorial-list.html}.}

We begin with a simple algebraic data type: binary trees, \AD{Tree}, with
two constructors.  \AC{leaf} takes no arguments, while \AC{node} takes two
arguments, both of type \AF{Tree}.  Functions on such data types are
typically defined by recursion and pattern matching; for example:
\begin{mathpar}
\codeblock{\begin{code}%
\>[2]\AgdaKeyword{data}\AgdaSpace{}%
\AgdaDatatype{Tree}\AgdaSpace{}%
\AgdaSymbol{:}\AgdaSpace{}%
\AgdaPrimitive{Set}\AgdaSpace{}%
\AgdaKeyword{where}\<%
\\
\>[2][@{}l@{\AgdaIndent{0}}]%
\>[4]\AgdaInductiveConstructor{leaf}%
\>[10]\AgdaSymbol{:}\AgdaSpace{}%
\AgdaDatatype{Tree}\<%
\\
\>[4]\AgdaInductiveConstructor{node}%
\>[10]\AgdaSymbol{:}\AgdaSpace{}%
\AgdaDatatype{Tree}\AgdaSpace{}%
\AgdaSymbol{→}\AgdaSpace{}%
\AgdaDatatype{Tree}\AgdaSpace{}%
\AgdaSymbol{→}\AgdaSpace{}%
\AgdaDatatype{Tree}\<%
\end{code}
}
\and
\codeblock{\begin{code}%
\>[2]\AgdaFunction{size}\AgdaSpace{}%
\AgdaSymbol{:}\AgdaSpace{}%
\AgdaDatatype{Tree}\AgdaSpace{}%
\AgdaSymbol{→}\AgdaSpace{}%
\AgdaDatatype{ℕ}\<%
\\
\>[2]\AgdaFunction{size}\AgdaSpace{}%
\AgdaInductiveConstructor{leaf}%
\>[19]\AgdaSymbol{=}\AgdaSpace{}%
\AgdaNumber{1}\<%
\\
\>[2]\AgdaFunction{size}\AgdaSpace{}%
\AgdaSymbol{(}\AgdaInductiveConstructor{node}\AgdaSpace{}%
\AgdaBound{l}\AgdaSpace{}%
\AgdaBound{r}\AgdaSymbol{)}%
\>[19]\AgdaSymbol{=}\AgdaSpace{}%
\AgdaNumber{1}\AgdaSpace{}%
\AgdaOperator{\AgdaPrimitive{+}}\AgdaSpace{}%
\AgdaFunction{size}\AgdaSpace{}%
\AgdaBound{l}\AgdaSpace{}%
\AgdaOperator{\AgdaPrimitive{+}}\AgdaSpace{}%
\AgdaFunction{size}\AgdaSpace{}%
\AgdaBound{r}\<%
\end{code}
}
\end{mathpar}
Under the Curry–Howard correspondence one may view \AF{Tree} as a
primitive proposition whose values are its proofs; in our translation,
however, such simple data types will correspond to terms rather than to
logical propositions.

The key feature of Agda is dependent types.  The following \AF{Path} type is
indexed by a \AF{Tree} and represents all valid paths in the given tree:
\begin{mathpar}
\codeblock{\begin{code}%
\>[2]\AgdaKeyword{data}\AgdaSpace{}%
\AgdaDatatype{Path}\AgdaSpace{}%
\AgdaSymbol{:}\AgdaSpace{}%
\AgdaDatatype{Tree}\AgdaSpace{}%
\AgdaSymbol{→}\AgdaSpace{}%
\AgdaPrimitive{Set}\AgdaSpace{}%
\AgdaKeyword{where}\<%
\\
\>[2][@{}l@{\AgdaIndent{0}}]%
\>[4]\AgdaInductiveConstructor{done}%
\>[11]\AgdaSymbol{:}\AgdaSpace{}%
\AgdaSymbol{∀}\AgdaSpace{}%
\AgdaSymbol{\{}\AgdaBound{t}\AgdaSymbol{\}}%
\>[22]\AgdaSymbol{→}\AgdaSpace{}%
\AgdaDatatype{Path}\AgdaSpace{}%
\AgdaBound{t}\<%
\\
\>[4]\AgdaInductiveConstructor{left}%
\>[11]\AgdaSymbol{:}\AgdaSpace{}%
\AgdaSymbol{∀}\AgdaSpace{}%
\AgdaSymbol{\{}\AgdaBound{l}\AgdaSpace{}%
\AgdaBound{r}\AgdaSymbol{\}}%
\>[22]\AgdaSymbol{→}\AgdaSpace{}%
\AgdaDatatype{Path}\AgdaSpace{}%
\AgdaBound{l}\AgdaSpace{}%
\AgdaSymbol{→}\AgdaSpace{}%
\AgdaDatatype{Path}\AgdaSpace{}%
\AgdaSymbol{(}\AgdaInductiveConstructor{node}\AgdaSpace{}%
\AgdaBound{l}\AgdaSpace{}%
\AgdaBound{r}\AgdaSymbol{)}\<%
\\
\>[4]\AgdaInductiveConstructor{right}%
\>[11]\AgdaSymbol{:}\AgdaSpace{}%
\AgdaSymbol{∀}\AgdaSpace{}%
\AgdaSymbol{\{}\AgdaBound{l}\AgdaSpace{}%
\AgdaBound{r}\AgdaSymbol{\}}%
\>[22]\AgdaSymbol{→}\AgdaSpace{}%
\AgdaDatatype{Path}\AgdaSpace{}%
\AgdaBound{r}\AgdaSpace{}%
\AgdaSymbol{→}\AgdaSpace{}%
\AgdaDatatype{Path}\AgdaSpace{}%
\AgdaSymbol{(}\AgdaInductiveConstructor{node}\AgdaSpace{}%
\AgdaBound{l}\AgdaSpace{}%
\AgdaBound{r}\AgdaSymbol{)}\<%
\end{code}
}
\and
\codeblock{\begin{code}%
\>[2]\AgdaFunction{length}\AgdaSpace{}%
\AgdaSymbol{:}\AgdaSpace{}%
\AgdaSymbol{∀}\AgdaSpace{}%
\AgdaSymbol{\{}\AgdaBound{t}\AgdaSymbol{\}}\AgdaSpace{}%
\AgdaSymbol{→}\AgdaSpace{}%
\AgdaDatatype{Path}\AgdaSpace{}%
\AgdaBound{t}\AgdaSpace{}%
\AgdaSymbol{→}\AgdaSpace{}%
\AgdaDatatype{ℕ}\<%
\\
\>[2]\AgdaFunction{length}\AgdaSpace{}%
\AgdaInductiveConstructor{done}%
\>[20]\AgdaSymbol{=}\AgdaSpace{}%
\AgdaNumber{0}\<%
\\
\>[2]\AgdaFunction{length}\AgdaSpace{}%
\AgdaSymbol{(}\AgdaInductiveConstructor{left}\AgdaSpace{}%
\AgdaBound{p}\AgdaSymbol{)}%
\>[20]\AgdaSymbol{=}\AgdaSpace{}%
\AgdaNumber{1}\AgdaSpace{}%
\AgdaOperator{\AgdaPrimitive{+}}\AgdaSpace{}%
\AgdaFunction{length}\AgdaSpace{}%
\AgdaBound{p}\<%
\\
\>[2]\AgdaFunction{length}\AgdaSpace{}%
\AgdaSymbol{(}\AgdaInductiveConstructor{right}\AgdaSpace{}%
\AgdaBound{p}\AgdaSymbol{)}%
\>[20]\AgdaSymbol{=}\AgdaSpace{}%
\AgdaNumber{1}\AgdaSpace{}%
\AgdaOperator{\AgdaPrimitive{+}}\AgdaSpace{}%
\AgdaFunction{length}\AgdaSpace{}%
\AgdaBound{p}\<%
\end{code}
}
\end{mathpar}
\AD{Path} has three constructors that constrain the shape of its index:
\AC{done}, which signals termination at the tree \AB{t}; and \AC{left} and
\AC{right}, which descend into the left or right subtree respectively,
when such subtrees exist.  The quantifier ∀ allows omission of explicit
types for the quantified variables; arguments in braces are implicit and
are not supplied during application or pattern matching.  The dependent
function \AF{length} computes over \AF{Path} values in the same manner as
over ordinary algebraic data types.

Under Curry–Howard, dependent types such as \AF{Path} correspond to
predicates: the predicate holds precisely when the corresponding type is
inhabited.  A noteworthy distinction is that some Agda predicates are
proof‑relevant (they may have many distinct inhabitants), whereas logical
propositions are typically proof‑irrelevant.  Consequently, types like
\AF{Path} admit multiple inhabitants that a function such as \AF{length}
may distinguish; this behaviour has no exact analogue in traditional
predicate logic.

When a predicate is propositional (that is, it has at most one
inhabitant), the logical proof and the Agda programme are interchangeable.
We confine our attention in this work to such propositions.  As an
example, consider \AD{IsEven} and propositional equality:
\begin{mathpar}
\codeblock{\begin{code}%
\>[2]\AgdaKeyword{data}\AgdaSpace{}%
\AgdaDatatype{IsEven}\AgdaSpace{}%
\AgdaSymbol{:}\AgdaSpace{}%
\AgdaDatatype{ℕ}\AgdaSpace{}%
\AgdaSymbol{→}\AgdaSpace{}%
\AgdaPrimitive{Set}\AgdaSpace{}%
\AgdaKeyword{where}\<%
\\
\>[2][@{}l@{\AgdaIndent{0}}]%
\>[4]\AgdaInductiveConstructor{zero}%
\>[10]\AgdaSymbol{:}\AgdaSpace{}%
\AgdaDatatype{IsEven}\AgdaSpace{}%
\AgdaNumber{0}\<%
\\
\>[4]\AgdaInductiveConstructor{suc}%
\>[10]\AgdaSymbol{:}\AgdaSpace{}%
\AgdaSymbol{∀}\AgdaSpace{}%
\AgdaSymbol{\{}\AgdaBound{n}\AgdaSymbol{\}}\AgdaSpace{}%
\AgdaSymbol{→}\AgdaSpace{}%
\AgdaDatatype{IsEven}\AgdaSpace{}%
\AgdaBound{n}\AgdaSpace{}%
\AgdaSymbol{→}\AgdaSpace{}%
\AgdaDatatype{IsEven}\AgdaSpace{}%
\AgdaSymbol{(}\AgdaNumber{2}\AgdaSpace{}%
\AgdaOperator{\AgdaPrimitive{+}}\AgdaSpace{}%
\AgdaBound{n}\AgdaSymbol{)}\<%
\end{code}
}
\and
\codeblock{\begin{code}%
\>[2]\AgdaKeyword{data}\AgdaSpace{}%
\AgdaOperator{\AgdaDatatype{\AgdaUnderscore{}≡\AgdaUnderscore{}}}\AgdaSpace{}%
\AgdaSymbol{\{}\AgdaBound{A}\AgdaSpace{}%
\AgdaSymbol{:}\AgdaSpace{}%
\AgdaPrimitive{Set}\AgdaSymbol{\}}\AgdaSpace{}%
\AgdaSymbol{(}\AgdaBound{x}\AgdaSpace{}%
\AgdaSymbol{:}\AgdaSpace{}%
\AgdaBound{A}\AgdaSymbol{)}\AgdaSpace{}%
\AgdaSymbol{:}\AgdaSpace{}%
\AgdaBound{A}\AgdaSpace{}%
\AgdaSymbol{→}\AgdaSpace{}%
\AgdaPrimitive{Set}\<%
\\
\>[2][@{}l@{\AgdaIndent{0}}]%
\>[3]\AgdaKeyword{where}\<%
\\
\>[3][@{}l@{\AgdaIndent{0}}]%
\>[4]\AgdaInductiveConstructor{refl}\AgdaSpace{}%
\AgdaSymbol{:}\AgdaSpace{}%
\AgdaBound{x}\AgdaSpace{}%
\AgdaOperator{\AgdaDatatype{≡}}\AgdaSpace{}%
\AgdaBound{x}\<%
\end{code}
}
\end{mathpar}
The type \AD{\_≡\_} asserts that two terms are equal when they reduce to
the same normal form; we translate this type as logical equality.

Finally, Agda provides a reflection mechanism, which yields a syntactic
representation (an abstract syntax tree) of internal terms --- a facility
analogous to quotation in Lisp.  Reflection was introduced to support
proof automation.  Since Agda lacks a separate tactic language, as found
in Rocq or Lean, reflection permits tactics to be implemented directly in
Agda.  One can quote a term to obtain its AST, manipulate the AST as
required, and then unquote the AST to produce an Agda term; the ensuing
unquotation invokes type checking, thereby ensuring that only well‑typed
terms are constructed.  The sequel illustrates how reflection is used to
construct a function \AF{foo} from an AST so that it coincides exactly
with a specified reference definition \AF{foo′}.
\vskip 5pt 
\noindent
\begin{code}[hide]%
\>[0]\AgdaKeyword{module}\AgdaSpace{}%
\AgdaModule{R}\AgdaSpace{}%
\AgdaKeyword{where}\<%
\\
\>[0][@{}l@{\AgdaIndent{0}}]%
\>[2]\AgdaKeyword{open}\AgdaSpace{}%
\AgdaKeyword{import}\AgdaSpace{}%
\AgdaModule{Data.Nat}\<%
\\
\>[2]\AgdaKeyword{open}\AgdaSpace{}%
\AgdaKeyword{import}\AgdaSpace{}%
\AgdaModule{Data.List}\<%
\\
\>[2]\AgdaKeyword{open}\AgdaSpace{}%
\AgdaKeyword{import}\AgdaSpace{}%
\AgdaModule{Data.Product}\<%
\\
\>[2]\AgdaKeyword{open}\AgdaSpace{}%
\AgdaKeyword{import}\AgdaSpace{}%
\AgdaModule{Data.Unit}\<%
\\
\>[2]\AgdaKeyword{open}\AgdaSpace{}%
\AgdaKeyword{import}\AgdaSpace{}%
\AgdaModule{Reflection}\<%
\\
\>[2]\AgdaKeyword{open}\AgdaSpace{}%
\AgdaModule{Clause}\<%
\\
\>[2]\AgdaKeyword{open}\AgdaSpace{}%
\AgdaModule{Pattern}\<%
\end{code}
\codeblock{\begin{code}%
\>[2]\AgdaKeyword{pattern}\AgdaSpace{}%
\AgdaInductiveConstructor{`ℕ}%
\>[20]\AgdaSymbol{=}\AgdaSpace{}%
\AgdaInductiveConstructor{def}\AgdaSpace{}%
\AgdaSymbol{(}\AgdaKeyword{quote}\AgdaSpace{}%
\AgdaDatatype{ℕ}\AgdaSymbol{)}\AgdaSpace{}%
\AgdaInductiveConstructor{[]}\<%
\\
\>[2]\AgdaKeyword{pattern}\AgdaSpace{}%
\AgdaInductiveConstructor{`zero}%
\>[20]\AgdaSymbol{=}\AgdaSpace{}%
\AgdaInductiveConstructor{con}\AgdaSpace{}%
\AgdaSymbol{(}\AgdaKeyword{quote}\AgdaSpace{}%
\AgdaInductiveConstructor{zero}\AgdaSymbol{)}\AgdaSpace{}%
\AgdaInductiveConstructor{[]}\<%
\end{code}
}
\codeblock{\begin{code}%
\>[2]\AgdaKeyword{pattern}\AgdaSpace{}%
\AgdaInductiveConstructor{`suc}\AgdaSpace{}%
\AgdaBound{x}%
\>[20]\AgdaSymbol{=}\AgdaSpace{}%
\AgdaInductiveConstructor{con}\AgdaSpace{}%
\AgdaSymbol{(}\AgdaKeyword{quote}\AgdaSpace{}%
\AgdaInductiveConstructor{suc}\AgdaSymbol{)}\AgdaSpace{}%
\AgdaSymbol{(}\AgdaInductiveConstructor{vArg}\AgdaSpace{}%
\AgdaBound{x}\AgdaSpace{}%
\AgdaOperator{\AgdaInductiveConstructor{∷}}\AgdaSpace{}%
\AgdaInductiveConstructor{[]}\AgdaSymbol{)}\<%
\\
\>[2]\AgdaKeyword{pattern}\AgdaSpace{}%
\AgdaOperator{\AgdaInductiveConstructor{\AgdaUnderscore{}`+\AgdaUnderscore{}}}\AgdaSpace{}%
\AgdaBound{x}\AgdaSpace{}%
\AgdaBound{y}%
\>[20]\AgdaSymbol{=}\AgdaSpace{}%
\AgdaInductiveConstructor{def}\AgdaSpace{}%
\AgdaSymbol{(}\AgdaKeyword{quote}\AgdaSpace{}%
\AgdaOperator{\AgdaPrimitive{\AgdaUnderscore{}+\AgdaUnderscore{}}}\AgdaSymbol{)}\AgdaSpace{}%
\AgdaSymbol{(}\AgdaInductiveConstructor{vArg}\AgdaSpace{}%
\AgdaBound{x}\AgdaSpace{}%
\AgdaOperator{\AgdaInductiveConstructor{∷}}\AgdaSpace{}%
\AgdaInductiveConstructor{vArg}\AgdaSpace{}%
\AgdaBound{y}\AgdaSpace{}%
\AgdaOperator{\AgdaInductiveConstructor{∷}}\AgdaSpace{}%
\AgdaInductiveConstructor{[]}\AgdaSymbol{)}\<%
\end{code}
}

We first define four \AK{patterns}, which resemble function definitions
but may be used within pattern‑matching expressions, and then use them to
construct the AST shown below:
\vskip 5pt 
\codeblock{\begin{code}%
\>[2]\AgdaFunction{foo}\AgdaSpace{}%
\AgdaSymbol{:}\AgdaSpace{}%
\AgdaDatatype{ℕ}\AgdaSpace{}%
\AgdaSymbol{→}\AgdaSpace{}%
\AgdaDatatype{ℕ}\<%
\\
\>[2]\AgdaKeyword{unquoteDef}\AgdaSpace{}%
\AgdaFunction{foo}\AgdaSpace{}%
\AgdaSymbol{=}\AgdaSpace{}%
\AgdaPostulate{defineFun}\AgdaSpace{}%
\AgdaFunction{foo}\<%
\\
\>[2][@{}l@{\AgdaIndent{0}}]%
\>[5]\AgdaSymbol{(}\AgdaSpace{}%
\AgdaInductiveConstructor{clause}\AgdaSpace{}%
\AgdaInductiveConstructor{[]}\AgdaSpace{}%
\AgdaSymbol{(}\AgdaInductiveConstructor{vArg}\AgdaSpace{}%
\AgdaInductiveConstructor{`zero}\AgdaSpace{}%
\AgdaOperator{\AgdaInductiveConstructor{∷}}\AgdaSpace{}%
\AgdaInductiveConstructor{[]}\AgdaSymbol{)}\AgdaSpace{}%
\AgdaInductiveConstructor{`zero}\<%
\\
\>[5]\AgdaOperator{\AgdaInductiveConstructor{∷}}\AgdaSpace{}%
\AgdaInductiveConstructor{clause}\AgdaSpace{}%
\AgdaSymbol{((}\AgdaString{"x"}\AgdaSpace{}%
\AgdaOperator{\AgdaInductiveConstructor{,}}\AgdaSpace{}%
\AgdaInductiveConstructor{vArg}\AgdaSpace{}%
\AgdaInductiveConstructor{`ℕ}\AgdaSymbol{)}\AgdaSpace{}%
\AgdaOperator{\AgdaInductiveConstructor{∷}}\AgdaSpace{}%
\AgdaInductiveConstructor{[]}\AgdaSymbol{)}%
\>[189I]\AgdaSymbol{(}\AgdaInductiveConstructor{vArg}\AgdaSpace{}%
\AgdaSymbol{(}\AgdaInductiveConstructor{`suc}\AgdaSpace{}%
\AgdaSymbol{(}\AgdaInductiveConstructor{var}\AgdaSpace{}%
\AgdaNumber{0}\AgdaSymbol{))}\AgdaSpace{}%
\AgdaOperator{\AgdaInductiveConstructor{∷}}\AgdaSpace{}%
\AgdaInductiveConstructor{[]}\AgdaSymbol{)}\<%
\\
\>[.][@{}l@{}]\<[189I]%
\>[37]\AgdaSymbol{(}\AgdaInductiveConstructor{var}\AgdaSpace{}%
\AgdaNumber{0}\AgdaSpace{}%
\AgdaInductiveConstructor{[]}\AgdaSpace{}%
\AgdaOperator{\AgdaInductiveConstructor{`+}}\AgdaSpace{}%
\AgdaInductiveConstructor{var}\AgdaSpace{}%
\AgdaNumber{0}\AgdaSpace{}%
\AgdaInductiveConstructor{[]}\AgdaSymbol{)}\AgdaSpace{}%
\AgdaOperator{\AgdaInductiveConstructor{∷}}\AgdaSpace{}%
\AgdaInductiveConstructor{[]}\AgdaSpace{}%
\AgdaSymbol{)}\<%
\end{code}
}
\codeblock{\begin{code}%
\>[2]\AgdaFunction{foo′}\AgdaSpace{}%
\AgdaSymbol{:}\AgdaSpace{}%
\AgdaDatatype{ℕ}\AgdaSpace{}%
\AgdaSymbol{→}\AgdaSpace{}%
\AgdaDatatype{ℕ}\<%
\\
\>[2]\AgdaFunction{foo′}\AgdaSpace{}%
\AgdaInductiveConstructor{zero}%
\>[15]\AgdaSymbol{=}\AgdaSpace{}%
\AgdaInductiveConstructor{zero}\<%
\\
\>[2]\AgdaFunction{foo′}\AgdaSpace{}%
\AgdaSymbol{(}\AgdaInductiveConstructor{suc}\AgdaSpace{}%
\AgdaBound{x}\AgdaSymbol{)}\AgdaSpace{}%
\AgdaSymbol{=}\AgdaSpace{}%
\AgdaBound{x}\AgdaSpace{}%
\AgdaOperator{\AgdaPrimitive{+}}\AgdaSpace{}%
\AgdaBound{x}\<%
\end{code}
}

\subsection{The \vampire Automated Theorem Prover}
\vampire~\cite{vampire,vampire-diary} is a successful automated theorem prover for classical first-order logic with equality.
It also has extensions for reasoning about fixed pre-defined \emph{theories} (such as arithmetic~\cite{alasca}) and for higher-order logic~\cite{vampire-ho}, but we use only the first-order core and support for inductively-defined data~\cite{coming-to-terms} here.
We chose \vampire because it is empirically very successful, and because one author was familiar with its behaviour, but we expect that this work could be profitably transplanted on to another ATP system that uses a similar calculus, including E~\cite{E}, Zipperposition~\cite{zipperposition}, or iProver~\cite{iProver}.
Unit equational systems like Twee~\cite{twee} or Waldmeister~\cite{waldmeister} could also be employed by translating Horn clauses~\cite{horn2ueq}.
Different systems will have complementary strengths.
\vampire also has potential for further extensions (\cref{sec:discussion}), improving its ability to dispatch Agda goals.
Given the audience, we would like to make a few clarifications about \vampire and systems like it:
\begin{description}
	\item[Fully automatic.] \vampire is a completely automatic push-button system that requires no interaction from the user. Given an input, it terminates only to report that the input is a theorem or non-theorem\footnote{in some (incomplete) configurations, \vampire may also give up}, or if a user-specified resource bound is exceeded.
\item[No magic.] The validity of a statement in first-order logic is undecidable in theory and hard in practice. \vampire is very good at proving theorems and can even disprove some non-theorems, but still may not terminate in reasonable time.
\item[Classical first-order logic.] \vampire is a classical system. The law of the excluded middle and other dicta of classical logic --- such as the non-emptiness of domains --- are baked very deeply into the system and cannot be disabled. Similarly, \vampire draws a hard line between \emph{terms}, which are spoken about, and \emph{formulae}, which do the speaking.
\item[Trust.] \vampire has no proof terms, kernel, or similar safety mechanism, and its proofs are simply a series of steps in its internal calculus that should be treated as \emph{untrusted}.
\end{description}

\noindent
\vampire, like other saturation-based ATPs, works by refutation.
Given an input containing axioms and a conjecture, the conjecture is negated and added to the axioms to form an input set of formulae.
These formulae are then preprocessed into \emph{clause normal form} (CNF), a set of clauses.
\emph{Search} now begins: \vampire applies the inference rules of its calculus to derive more and more consequences of the clauses, working fairly so that a proof --- if it exists --- will eventually be found.
When \vampire derives the empty clause (falsum), it has refuted the negation of the input conjecture.
Classically, this is the same as proving the conjecture.

\vampire implements the \emph{superposition} calculus~\cite{paramodulation}, a complete calculus for classical first-order logic with special treatment for the equality predicate.
For the sake of efficiency the calculus has many restrictions on rules, but for our purposes this can be ignored.
The relevant inference rules are as follows:
\begin{mathpar}
\inferrule{L \vee C \\ \lnot L' \vee D}{\sigma(C \vee D)}\quad\textsc{Resolution}
\and
\inferrule{L \vee L' \vee C}{\sigma(L \vee C)}\quad\textsc{Factoring}
\and
\inferrule{l = r \vee C \\ L[l'] \vee D}{\sigma({L[r] \vee C \vee D})}\quad\textsc{Superposition}
\and
\inferrule{s \neq t \vee C}{\sigma(C)}\quad\textsc{Equality Resolution}
\end{mathpar}
where the substitution $\sigma$ is the most general unifier of $L$ and $L'$ in \textsc{Resolution} and \textsc{Factoring}; $l$ and $l'$ in \textsc{Superposition}; and $s$, $t$ in \textsc{Equality Resolution}.

To build intuition for these rules, it is helpful to imagine that the unified terms and literals are already identical in the premises.
\textsc{Resolution} can then be seen as case analysis on the Boolean $L$, \textsc{Factoring} as an application of the identity $(p \vee p) \equiv p$, \textsc{Superposition} as conditional rewriting, and \textsc{Equality Resolution} as removing an obviously-false literal $t \neq t$ from a clause.
The full version of each rule can be seen as first (partially) instantiating universal quantifiers according to $\sigma$, then applying the rule without unification.

\begin{example}[resolution]
The following concrete inference is an instance of the resolution rule. The most general unifier of $P(x, f(y))$ and $P(c, z))$ is $\{x \mapsto c,~z \mapsto f(y)\}$.
\begin{mathpar}
\inferrule{P(x, f(y)) \vee Q(y) \\ \lnot P(c, z) \vee \lnot R(z)}{Q(y) \vee \lnot R(f(y))}\quad\textsc{Resolution}
\end{mathpar}
\end{example}

\noindent
There are many other rules which we deal with later (\cref{sec:edge-cases}), but the above four rules form the core first-order reasoning of \vampire.
\vampire also has a type system~\cite{tff1,polymorphic-vampire} for terms: it is essentially System F restricted to rank-1 types.

\section{Method}
Using \vampire to dispatch Agda goals does not seem to be workable at first sight.
\vampire cannot process everything in Agda's vocabulary, and Agda will not accept \vampire's classical proofs, even if we get one.
Agda also lacks the automation required to do Sledgehammer-style proof reconstruction~\cite{sledgehammer} from a set of premises minimised by \vampire.
However, by restricting \emph{both} systems to a common language, we can give Agda goals directly to \vampire and translate \vampire's proofs directly to an Agda proof term.
\begin{tcolorbox}
\textbf{Observation.}
Many useful Agda premises and goals are of the form
$$\forall \bar{x}.~R_1 \to R_2 \to \ldots \to R_n$$
where every $R_i$ is an atomic type. This corresponds to the Horn fragment of FOL.
\end{tcolorbox}
\noindent
There are more details involved in translating an Agda goal into Horn logic (\cref{sec:agda-to-vampire}), but suppose for now that we have a set of definite clauses $\Gamma$ and a ground ``goal'' atom $G$, and we wish to show $\Gamma \vdash G$.
Assuming that this is provable, and that \vampire manages to prove it, we obtain a classical proof by refutation showing $\Gamma, \lnot G \vdash \bot$.
But we would simply like $\Gamma \vdash G$, with no extra negations.
This leads us to another observation.
\begin{tcolorbox}
\textbf{Observation.}
\vampire's inference rules conclude Horn clauses if the premises are also Horn.
Furthermore, \vampire's Horn proofs remain valid if goal clauses $C \to \bot$ are systematically replaced with $C \to G$, for some ground $G$.
\end{tcolorbox}
\noindent
This allows us to transform the proof mechanically, replacing occurrences of $\bot$ with $G$.
The result is a classical proof showing $\Gamma, (G \to G) \vdash G$, but $G \to G$ is trivial and we obtain $\Gamma \vdash G$ immediately.
But we still require a constructive proof.
\begin{tcolorbox}
\textbf{Observation.}
\vampire's core calculus is constructively valid.
\end{tcolorbox}
\noindent
The classical aspects of \vampire are confined to the normal form translation (which may e.g. perform double-negation elimination), and the principle of proof by refutation.
Overall, this means that an Agda hole $G$ in context $\Gamma$ can be translated into an input that \vampire can attempt (\cref{sec:agda-to-vampire}), and the resulting proof can be transformed (\cref{sec:proof-transform}) into a constructive proof of $\Gamma \vdash G$ as Agda terms (\cref{sec:vampire-to-agda}), automatically filling the hole.
We will now be precise about the details of this round-trip.

\begin{code}[hide]%
\>[0]\AgdaComment{--\ Uncomment\ for\ interactive\ editing\ with\ no\ extra\ flags.}\<%
\\
\>[0]\AgdaComment{--\ module\ papers.ITP26.agda-to-vampire\ where}\<%
\\
\\%
\>[0]\AgdaKeyword{open}\AgdaSpace{}%
\AgdaKeyword{import}\AgdaSpace{}%
\AgdaModule{Relation.Binary.PropositionalEquality}\AgdaSpace{}%
\AgdaKeyword{hiding}\AgdaSpace{}%
\AgdaSymbol{(}\AgdaOperator{\AgdaInductiveConstructor{[\AgdaUnderscore{}]}}\AgdaSymbol{)}\<%
\\
\>[0]\AgdaKeyword{open}\AgdaSpace{}%
\AgdaKeyword{import}\AgdaSpace{}%
\AgdaModule{Data.String}\AgdaSpace{}%
\AgdaKeyword{using}\AgdaSpace{}%
\AgdaSymbol{(}\AgdaPostulate{String}\AgdaSymbol{)}\<%
\\
\>[0]\AgdaKeyword{open}\AgdaSpace{}%
\AgdaKeyword{import}\AgdaSpace{}%
\AgdaModule{Data.List}\<%
\\
\>[0]\AgdaKeyword{open}\AgdaSpace{}%
\AgdaKeyword{import}\AgdaSpace{}%
\AgdaModule{Extract}\<%
\\
\>[0]\AgdaKeyword{open}\AgdaSpace{}%
\AgdaKeyword{import}\AgdaSpace{}%
\AgdaModule{Smt}\<%
\end{code}

\subsection{\label{sec:agda-to-vampire}Translating Agda Terms to Horn Clauses}

The translation of Agda formulae proceeds in three conceptual stages:
(i) reflecting the user‑defined definitions;
(ii) checking that these definitions lie within the supported language fragment; and
(iii) converting the reflected terms to an input language \vampire supports.
We chose SMT-LIB~\cite{smtlib} as it can express datatypes.
Unlike the strict fragment characterisation employed by Foster and Struth~\cite{struth-agda-atp},
we prioritise ease of use.  Consequently we rely on Agda's reflection
facilities rather than on labour‑intensive manual encodings.  This apparent
laxity is justified as the overall process of finding a valid proof may
fail.  There is no guarantee that the
back‑end ATP will discover a proof; even if a proof is found, the
translated proof may be unsound, which is why the imported proof must be
type‑checked.  Moreover, as our integration evolves,
this modular approach facilitates experimentation with new features.

\paragraph*{Reflection}
We provide a collection of macros that accept the name of an Agda definition
quote the term and emit an SMT-LIB representation.  In particular:
\begin{description}
\item[\AM{kty}] declares a new sort to \vampire.
\item[\AM{kdata}] translates an inductive algebraic
data type (without parameters or indices) into a built‑in datatype declaration,
handling constructors and ensuring well‑foundedness, thereby enabling inductive reasoning.
\item[\AM{kfun}] converts functions given by pattern‑matching equations
into a function symbol with the appropriate type signature, together with
clauses that become axioms.
\end{description}
To illustrate, consider addition on Peano‑encoded natural numbers:
\begin{code}[hide]%
\>[0]\AgdaKeyword{module}\AgdaSpace{}%
\AgdaModule{\AgdaUnderscore{}}\AgdaSpace{}%
\AgdaKeyword{where}\<%
\\
\>[0]\AgdaKeyword{module}\AgdaSpace{}%
\AgdaModule{Ex}\AgdaSpace{}%
\AgdaKeyword{where}\<%
\end{code}
\begin{code}%
\>[0][@{}l@{\AgdaIndent{1}}]%
\>[2]\AgdaKeyword{data}\AgdaSpace{}%
\AgdaDatatype{Nat}\AgdaSpace{}%
\AgdaSymbol{:}\AgdaSpace{}%
\AgdaPrimitive{Set}\AgdaSpace{}%
\AgdaKeyword{where}%
\>[33]\AgdaComment{--\ (declare-datatype\ nat\ }\<%
\\
\>[2][@{}l@{\AgdaIndent{0}}]%
\>[4]\AgdaInductiveConstructor{ze}\AgdaSpace{}%
\AgdaSymbol{:}\AgdaSpace{}%
\AgdaDatatype{Nat}%
\>[33]\AgdaComment{--\ \ ((ze)}\<%
\\
\>[4]\AgdaInductiveConstructor{su}\AgdaSpace{}%
\AgdaSymbol{:}\AgdaSpace{}%
\AgdaSymbol{(}\AgdaBound{n}\AgdaSpace{}%
\AgdaSymbol{:}\AgdaSpace{}%
\AgdaDatatype{Nat}\AgdaSymbol{)}\AgdaSpace{}%
\AgdaSymbol{→}\AgdaSpace{}%
\AgdaDatatype{Nat}%
\>[33]\AgdaComment{--\ \ \ (su\ (n\ nat))))\ }\<%
\\
\\%
\>[2]\AgdaFunction{add}\AgdaSpace{}%
\AgdaSymbol{:}\AgdaSpace{}%
\AgdaDatatype{Nat}\AgdaSpace{}%
\AgdaSymbol{→}\AgdaSpace{}%
\AgdaDatatype{Nat}\AgdaSpace{}%
\AgdaSymbol{→}\AgdaSpace{}%
\AgdaDatatype{Nat}%
\>[33]\AgdaComment{--\ (declare-fun\ add\ (nat\ nat)\ nat)}\<%
\\
\>[2]\AgdaFunction{add}\AgdaSpace{}%
\AgdaInductiveConstructor{ze}\AgdaSpace{}%
\AgdaBound{x}\AgdaSpace{}%
\AgdaSymbol{=}\AgdaSpace{}%
\AgdaBound{x}%
\>[33]\AgdaComment{--\ (assert\ (!\ (forall\ ((x\ nat))\ }\<%
\\
\>[33]\AgdaComment{--\ \ \ \ \ \ \ \ \ \ \ \ (=\ (add\ ze\ x)\ x))\ :named\ add-clause-1))}\<%
\\
\>[2]\AgdaFunction{add}\AgdaSpace{}%
\AgdaSymbol{(}\AgdaInductiveConstructor{su}\AgdaSpace{}%
\AgdaBound{x}\AgdaSymbol{)}\AgdaSpace{}%
\AgdaBound{y}\AgdaSpace{}%
\AgdaSymbol{=}\AgdaSpace{}%
\AgdaInductiveConstructor{su}\AgdaSpace{}%
\AgdaSymbol{(}\AgdaFunction{add}\AgdaSpace{}%
\AgdaBound{x}\AgdaSpace{}%
\AgdaBound{y}\AgdaSymbol{)}%
\>[33]\AgdaComment{--\ (assert\ (!\ (forall\ ((x\ nat)(y\ nat))}\<%
\\
\>[33]\AgdaComment{--\ \ \ \ \ \ \ \ \ \ \ \ (=\ (add\ (su\ x)\ y)\ (su\ (add\ x\ y))))}\<%
\\
\>[33]\AgdaComment{--\ \ :named\ add-clause-2))}\<%
\end{code}
Stated theorems (\eg{} some facts about natural numbers) can be turned
into an axiom using the \AM{kthm} macro; this simply translates the type into an assertion.
Note that we do not verify the existence of a proof --- we only inspect the type.
Finally, the \AM{kgoal} macro is used to translate goals that we wish
to prove.

\paragraph*{Conformity checks}
The macros above impose several constraints on the shape of the supplied
datatypes and functions.  If these constraints are violated, the translation
aborts with an error. Specifically:
\begin{itemize}
\item For inductively defined datatypes, parameters
or indices are not permitted.  Although \vampire offers limited support for
parametrised types such as polymorphic lists, we currently do not
handle them.
\item For functions we demand arguments of simple, non‑dependent types, thereby
forbiding dependent functions.
\item For theorems we enforce that the overall type resides in the Horn
fragment.  Consequently, the type may contain a mixture of named arguments of
simple datatypes and propositional components that reference variables of
simple types, while the right‑hand side must be a proposition.  At present the
only supported proposition is propositional equality\footnote{although \vampire and the proof conversion can handle arbitrary relations}, which implicitly means
that we assume Uniqueness of Identity Proofs (UIP) in Agda.
\item Goals, like theorems, must be in Horn form.  During translation we
convert all premises of the goal into named constants and negate the right‑hand
side using \texttt{assert‑not}\footnote{\texttt{(assert-not F)} is a Vampiric extension with the same semantics as \texttt{(assert (not F))}, but with the additional effect of marking \texttt{F} as a goal for \vampire's internal heuristics}.
\end{itemize}

\paragraph*{Translation}
We define a type to represent a subset of SMT-LIB programs.  Macros tranlsate
turn into the element of this type which is subsequently rendered 
as a string.  This string is fed to \vampire; the
resulting proof is processed by the proof‑converter described in
Section~\ref{sec:vampire-to-agda} and finally emitted as Agda code.

\paragraph*{Example}
To demonstrate the workflow, we consider a rudimentary theory of vector spaces
comprising a zero vector \AF{ze}, addition \AF{\_+\_}, and negation \AF{-\_},
together with three axioms: associativity of addition, left‑hand neutrality of
zero, and left‑hand inverse for addition:
\begin{mathpar}
\begin{code}[hide]%
\>[0]\AgdaKeyword{postulate}\<%
\end{code}
\codeblock{\begin{code}%
\>[0][@{}l@{\AgdaIndent{1}}]%
\>[2]\AgdaPostulate{V}%
\>[8]\AgdaSymbol{:}\AgdaSpace{}%
\AgdaPrimitive{Set}\<%
\\
\>[2]\AgdaPostulate{ze}%
\>[8]\AgdaSymbol{:}\AgdaSpace{}%
\AgdaPostulate{V}\<%
\\
\>[2]\AgdaOperator{\AgdaPostulate{\AgdaUnderscore{}+\AgdaUnderscore{}}}%
\>[8]\AgdaSymbol{:}\AgdaSpace{}%
\AgdaPostulate{V}\AgdaSpace{}%
\AgdaSymbol{→}\AgdaSpace{}%
\AgdaPostulate{V}\AgdaSpace{}%
\AgdaSymbol{→}\AgdaSpace{}%
\AgdaPostulate{V}\<%
\\
\>[2]\AgdaOperator{\AgdaPostulate{-\AgdaUnderscore{}}}%
\>[8]\AgdaSymbol{:}\AgdaSpace{}%
\AgdaPostulate{V}\AgdaSpace{}%
\AgdaSymbol{→}\AgdaSpace{}%
\AgdaPostulate{V}\<%
\end{code}
}
\and
\codeblock{\begin{code}%
\>[2]\AgdaPostulate{assoc}\AgdaSpace{}%
\AgdaSymbol{:}\AgdaSpace{}%
\AgdaSymbol{∀}\AgdaSpace{}%
\AgdaBound{u}\AgdaSpace{}%
\AgdaBound{v}\AgdaSpace{}%
\AgdaBound{w}\AgdaSpace{}%
\AgdaSymbol{→}\AgdaSpace{}%
\AgdaSymbol{(}\AgdaBound{u}\AgdaSpace{}%
\AgdaOperator{\AgdaPostulate{+}}\AgdaSpace{}%
\AgdaBound{v}\AgdaSymbol{)}\AgdaSpace{}%
\AgdaOperator{\AgdaPostulate{+}}\AgdaSpace{}%
\AgdaBound{w}\AgdaSpace{}%
\AgdaOperator{\AgdaDatatype{≡}}\AgdaSpace{}%
\AgdaBound{u}\AgdaSpace{}%
\AgdaOperator{\AgdaPostulate{+}}\AgdaSpace{}%
\AgdaSymbol{(}\AgdaBound{v}\AgdaSpace{}%
\AgdaOperator{\AgdaPostulate{+}}\AgdaSpace{}%
\AgdaBound{w}\AgdaSymbol{)}\<%
\\
\>[2]\AgdaPostulate{neutl}\AgdaSpace{}%
\AgdaSymbol{:}\AgdaSpace{}%
\AgdaSymbol{∀}\AgdaSpace{}%
\AgdaBound{u}\AgdaSpace{}%
\AgdaSymbol{→}\AgdaSpace{}%
\AgdaPostulate{ze}\AgdaSpace{}%
\AgdaOperator{\AgdaPostulate{+}}\AgdaSpace{}%
\AgdaBound{u}\AgdaSpace{}%
\AgdaOperator{\AgdaDatatype{≡}}\AgdaSpace{}%
\AgdaBound{u}\<%
\\
\>[2]\AgdaPostulate{negl}%
\>[8]\AgdaSymbol{:}\AgdaSpace{}%
\AgdaSymbol{∀}\AgdaSpace{}%
\AgdaBound{u}\AgdaSpace{}%
\AgdaSymbol{→}\AgdaSpace{}%
\AgdaSymbol{(}\AgdaOperator{\AgdaPostulate{-}}\AgdaSpace{}%
\AgdaBound{u}\AgdaSymbol{)}\AgdaSpace{}%
\AgdaOperator{\AgdaPostulate{+}}\AgdaSpace{}%
\AgdaBound{u}\AgdaSpace{}%
\AgdaOperator{\AgdaDatatype{≡}}\AgdaSpace{}%
\AgdaPostulate{ze}\<%
\end{code}
}
\end{mathpar}
These declarations suffice to state the uniqueness of the zero vector,
expressed by the following theorem (the proof will be obtained from \vampire):
\begin{code}%
\>[0]\AgdaFunction{ze-uniq}\AgdaSpace{}%
\AgdaSymbol{:}\AgdaSpace{}%
\AgdaSymbol{∀}\AgdaSpace{}%
\AgdaBound{u}\AgdaSpace{}%
\AgdaBound{v}\AgdaSpace{}%
\AgdaSymbol{→}\AgdaSpace{}%
\AgdaSymbol{(}\AgdaBound{e}\AgdaSpace{}%
\AgdaSymbol{:}\AgdaSpace{}%
\AgdaBound{u}\AgdaSpace{}%
\AgdaOperator{\AgdaPostulate{+}}\AgdaSpace{}%
\AgdaBound{v}\AgdaSpace{}%
\AgdaOperator{\AgdaDatatype{≡}}\AgdaSpace{}%
\AgdaBound{v}\AgdaSymbol{)}\AgdaSpace{}%
\AgdaSymbol{→}\AgdaSpace{}%
\AgdaBound{u}\AgdaSpace{}%
\AgdaOperator{\AgdaDatatype{≡}}\AgdaSpace{}%
\AgdaPostulate{ze}\<%
\end{code}
To obtain the goal we first extract the type \AF{V}, the types of the
functions \AF{\_+\_}, \AF{-\_} and the constant \AF{ze} (treated as a nullary
function).  We then supply the three axioms and finally specify the goal to be
proved.  As can be seen, only the names of the definitions need to be provided;
the rest of the conversion is performed automatically by the reflection machinery:
\begin{code}%
\>[0]\AgdaFunction{ze-uniq-v}\AgdaSpace{}%
\AgdaSymbol{:}\AgdaSpace{}%
\AgdaPostulate{String}\<%
\\
\>[0]\AgdaFunction{ze-uniq-v}\AgdaSpace{}%
\AgdaSymbol{=}\AgdaSpace{}%
\AgdaFunction{print-prog}%
\>[24]\AgdaSymbol{(}\AgdaMacro{kty}\AgdaSpace{}%
\AgdaPostulate{V}\AgdaSpace{}%
\AgdaOperator{\AgdaInductiveConstructor{∷}}\AgdaSpace{}%
\AgdaMacro{kfun-ty}\AgdaSpace{}%
\AgdaOperator{\AgdaPostulate{\AgdaUnderscore{}+\AgdaUnderscore{}}}\AgdaSpace{}%
\AgdaOperator{\AgdaInductiveConstructor{∷}}\AgdaSpace{}%
\AgdaMacro{kfun-ty}\AgdaSpace{}%
\AgdaOperator{\AgdaPostulate{-\AgdaUnderscore{}}}\AgdaSpace{}%
\AgdaOperator{\AgdaInductiveConstructor{∷}}\AgdaSpace{}%
\AgdaMacro{kfun-ty}\AgdaSpace{}%
\AgdaPostulate{ze}\AgdaSpace{}%
\AgdaOperator{\AgdaInductiveConstructor{∷}}\AgdaSpace{}%
\AgdaMacro{kthm}\AgdaSpace{}%
\AgdaPostulate{neutl}\<%
\\
\>[24][@{}l@{\AgdaIndent{0}}]%
\>[25]\AgdaOperator{\AgdaInductiveConstructor{∷}}\AgdaSpace{}%
\AgdaMacro{kthm}\AgdaSpace{}%
\AgdaPostulate{negl}%
\>[38]\AgdaOperator{\AgdaInductiveConstructor{∷}}\AgdaSpace{}%
\AgdaMacro{kthm}\AgdaSpace{}%
\AgdaPostulate{assoc}\AgdaSpace{}%
\AgdaOperator{\AgdaInductiveConstructor{∷}}\AgdaSpace{}%
\AgdaMacro{kgoal}\AgdaSpace{}%
\AgdaFunction{ze-uniq}\AgdaSpace{}%
\AgdaOperator{\AgdaInductiveConstructor{∷}}\AgdaSpace{}%
\AgdaInductiveConstructor{[]}\AgdaSymbol{)}\<%
\end{code}
Evaluating \AF{ze‑uniq‑v} yields the following SMT-LIB:
\begin{Verbatim}[fontsize=\small]
(declare-sort vec.V 0)
(declare-fun vec._+_ (vec.V vec.V) vec.V)
(declare-fun vec.-_ (vec.V) vec.V)
(declare-const vec.ze vec.V)
(assert (! (forall ((u vec.V)) (= (vec._+_ (vec.ze ) u) u)) :named vec.neutl))
(assert (! (forall ((u vec.V)) (= (vec._+_ (vec.-_ u) u) (vec.ze ))) :named vec.negl))
(assert (! (forall ((u vec.V) (v vec.V) (w vec.V))
                    (= (vec._+_ (vec._+_ u v) w) (vec._+_ u (vec._+_ v w))))
                    :named vec.assoc))
(declare-const u vec.V)
(declare-const v vec.V)
(assert (! (forall () (= (vec._+_ u v) v)) :named e))
(assert-not (! (= u (vec.ze )) :named vec.ze-uniq))
\end{Verbatim}
The required proof is non‑trivial and is not discovered by Agda's built‑in
automatic proof search.
\vampire, however, finds a correct proof instantly.

\begin{code}[hide]%
\>[0]\AgdaFunction{ze-uniq}\AgdaSpace{}%
\AgdaBound{u}\AgdaSpace{}%
\AgdaBound{v}\AgdaSpace{}%
\AgdaBound{e}\AgdaSpace{}%
\AgdaSymbol{=}\AgdaSpace{}%
\AgdaFunction{vampagda-f93}\AgdaSpace{}%
\AgdaKeyword{where}\<%
\\
\>[0][@{}l@{\AgdaIndent{0}}]%
\>[2]\AgdaFunction{vampagda-f1}\AgdaSpace{}%
\AgdaSymbol{:}\AgdaSpace{}%
\AgdaSymbol{∀}\AgdaSpace{}%
\AgdaSymbol{(}\AgdaBound{x0}\AgdaSpace{}%
\AgdaSymbol{:}\AgdaSpace{}%
\AgdaPostulate{V}\AgdaSymbol{)}\AgdaSpace{}%
\AgdaSymbol{→}\AgdaSpace{}%
\AgdaSymbol{(}\AgdaOperator{\AgdaPostulate{\AgdaUnderscore{}+\AgdaUnderscore{}}}\AgdaSpace{}%
\AgdaPostulate{ze}\AgdaSpace{}%
\AgdaBound{x0}\AgdaSymbol{)}\AgdaSpace{}%
\AgdaOperator{\AgdaDatatype{≡}}\AgdaSpace{}%
\AgdaBound{x0}\<%
\\
\>[2]\AgdaFunction{vampagda-f1}\AgdaSpace{}%
\AgdaBound{x0}\AgdaSpace{}%
\AgdaSymbol{=}\AgdaSpace{}%
\AgdaSymbol{(}\AgdaPostulate{neutl}\AgdaSpace{}%
\AgdaBound{x0}\AgdaSymbol{)}\<%
\\
\>[2]\AgdaFunction{vampagda-f2}\AgdaSpace{}%
\AgdaSymbol{:}\AgdaSpace{}%
\AgdaSymbol{∀}\AgdaSpace{}%
\AgdaSymbol{(}\AgdaBound{x0}\AgdaSpace{}%
\AgdaSymbol{:}\AgdaSpace{}%
\AgdaPostulate{V}\AgdaSymbol{)}\AgdaSpace{}%
\AgdaSymbol{→}\AgdaSpace{}%
\AgdaSymbol{(}\AgdaOperator{\AgdaPostulate{\AgdaUnderscore{}+\AgdaUnderscore{}}}\AgdaSpace{}%
\AgdaSymbol{(}\AgdaOperator{\AgdaPostulate{-\AgdaUnderscore{}}}\AgdaSpace{}%
\AgdaBound{x0}\AgdaSymbol{)}\AgdaSpace{}%
\AgdaBound{x0}\AgdaSymbol{)}\AgdaSpace{}%
\AgdaOperator{\AgdaDatatype{≡}}\AgdaSpace{}%
\AgdaPostulate{ze}\<%
\\
\>[2]\AgdaFunction{vampagda-f2}\AgdaSpace{}%
\AgdaBound{x0}\AgdaSpace{}%
\AgdaSymbol{=}\AgdaSpace{}%
\AgdaSymbol{(}\AgdaPostulate{negl}\AgdaSpace{}%
\AgdaBound{x0}\AgdaSymbol{)}\<%
\\
\>[2]\AgdaFunction{vampagda-f3}\AgdaSpace{}%
\AgdaSymbol{:}\AgdaSpace{}%
\AgdaSymbol{∀}\AgdaSpace{}%
\AgdaSymbol{(}\AgdaBound{x0}\AgdaSpace{}%
\AgdaSymbol{:}\AgdaSpace{}%
\AgdaPostulate{V}\AgdaSymbol{)}\AgdaSpace{}%
\AgdaSymbol{(}\AgdaBound{x1}\AgdaSpace{}%
\AgdaSymbol{:}\AgdaSpace{}%
\AgdaPostulate{V}\AgdaSymbol{)}\AgdaSpace{}%
\AgdaSymbol{(}\AgdaBound{x2}\AgdaSpace{}%
\AgdaSymbol{:}\AgdaSpace{}%
\AgdaPostulate{V}\AgdaSymbol{)}\AgdaSpace{}%
\AgdaSymbol{→}\AgdaSpace{}%
\AgdaSymbol{(}\AgdaOperator{\AgdaPostulate{\AgdaUnderscore{}+\AgdaUnderscore{}}}\AgdaSpace{}%
\AgdaSymbol{(}\AgdaOperator{\AgdaPostulate{\AgdaUnderscore{}+\AgdaUnderscore{}}}\AgdaSpace{}%
\AgdaBound{x0}\AgdaSpace{}%
\AgdaBound{x1}\AgdaSymbol{)}\AgdaSpace{}%
\AgdaBound{x2}\AgdaSymbol{)}\AgdaSpace{}%
\AgdaOperator{\AgdaDatatype{≡}}\AgdaSpace{}%
\AgdaSymbol{(}\AgdaOperator{\AgdaPostulate{\AgdaUnderscore{}+\AgdaUnderscore{}}}\AgdaSpace{}%
\AgdaBound{x0}\AgdaSpace{}%
\AgdaSymbol{(}\AgdaOperator{\AgdaPostulate{\AgdaUnderscore{}+\AgdaUnderscore{}}}\AgdaSpace{}%
\AgdaBound{x1}\AgdaSpace{}%
\AgdaBound{x2}\AgdaSymbol{))}\<%
\\
\>[2]\AgdaFunction{vampagda-f3}\AgdaSpace{}%
\AgdaBound{x0}\AgdaSpace{}%
\AgdaBound{x1}\AgdaSpace{}%
\AgdaBound{x2}\AgdaSpace{}%
\AgdaSymbol{=}\AgdaSpace{}%
\AgdaSymbol{(}\AgdaPostulate{assoc}\AgdaSpace{}%
\AgdaBound{x0}\AgdaSpace{}%
\AgdaBound{x1}\AgdaSpace{}%
\AgdaBound{x2}\AgdaSymbol{)}\<%
\\
\>[2]\AgdaFunction{vampagda-f4}\AgdaSpace{}%
\AgdaSymbol{:}\AgdaSpace{}%
\AgdaSymbol{(}\AgdaOperator{\AgdaPostulate{\AgdaUnderscore{}+\AgdaUnderscore{}}}\AgdaSpace{}%
\AgdaBound{u}\AgdaSpace{}%
\AgdaBound{v}\AgdaSymbol{)}\AgdaSpace{}%
\AgdaOperator{\AgdaDatatype{≡}}\AgdaSpace{}%
\AgdaBound{v}\<%
\\
\>[2]\AgdaFunction{vampagda-f4}\AgdaSpace{}%
\AgdaSymbol{=}\AgdaSpace{}%
\AgdaBound{e}\<%
\\
\>[2]\AgdaFunction{vampagda-f6}\AgdaSpace{}%
\AgdaSymbol{:}\AgdaSpace{}%
\AgdaBound{u}\AgdaSpace{}%
\AgdaOperator{\AgdaDatatype{≡}}\AgdaSpace{}%
\AgdaPostulate{ze}\AgdaSpace{}%
\AgdaSymbol{→}\AgdaSpace{}%
\AgdaBound{u}\AgdaSpace{}%
\AgdaOperator{\AgdaDatatype{≡}}\AgdaSpace{}%
\AgdaPostulate{ze}\<%
\\
\>[2]\AgdaFunction{vampagda-f6}\AgdaSpace{}%
\AgdaBound{l0}\AgdaSpace{}%
\AgdaSymbol{=}\AgdaSpace{}%
\AgdaBound{l0}\<%
\\
\>[2]\AgdaFunction{vampagda-f7}\AgdaSpace{}%
\AgdaSymbol{:}\AgdaSpace{}%
\AgdaPostulate{ze}\AgdaSpace{}%
\AgdaOperator{\AgdaDatatype{≡}}\AgdaSpace{}%
\AgdaBound{u}\AgdaSpace{}%
\AgdaSymbol{→}\AgdaSpace{}%
\AgdaBound{u}\AgdaSpace{}%
\AgdaOperator{\AgdaDatatype{≡}}\AgdaSpace{}%
\AgdaPostulate{ze}\<%
\\
\>[2]\AgdaFunction{vampagda-f7}\AgdaSpace{}%
\AgdaBound{l0}\AgdaSpace{}%
\AgdaSymbol{=}\AgdaSpace{}%
\AgdaSymbol{(}\AgdaFunction{vampagda-f6}\AgdaSpace{}%
\AgdaSymbol{(}\AgdaFunction{sym}\AgdaSpace{}%
\AgdaBound{l0}\AgdaSymbol{))}\<%
\\
\>[2]\AgdaFunction{vampagda-f8}\AgdaSpace{}%
\AgdaSymbol{:}\AgdaSpace{}%
\AgdaSymbol{∀}\AgdaSpace{}%
\AgdaSymbol{(}\AgdaBound{x0}\AgdaSpace{}%
\AgdaSymbol{:}\AgdaSpace{}%
\AgdaPostulate{V}\AgdaSymbol{)}\AgdaSpace{}%
\AgdaSymbol{→}\AgdaSpace{}%
\AgdaSymbol{(}\AgdaOperator{\AgdaPostulate{\AgdaUnderscore{}+\AgdaUnderscore{}}}\AgdaSpace{}%
\AgdaPostulate{ze}\AgdaSpace{}%
\AgdaBound{x0}\AgdaSymbol{)}\AgdaSpace{}%
\AgdaOperator{\AgdaDatatype{≡}}\AgdaSpace{}%
\AgdaBound{x0}\<%
\\
\>[2]\AgdaFunction{vampagda-f8}\AgdaSpace{}%
\AgdaBound{x0}\AgdaSpace{}%
\AgdaSymbol{=}\AgdaSpace{}%
\AgdaSymbol{(}\AgdaFunction{vampagda-f1}\AgdaSpace{}%
\AgdaBound{x0}\AgdaSymbol{)}\<%
\\
\>[2]\AgdaFunction{vampagda-f9}\AgdaSpace{}%
\AgdaSymbol{:}\AgdaSpace{}%
\AgdaSymbol{∀}\AgdaSpace{}%
\AgdaSymbol{(}\AgdaBound{x0}\AgdaSpace{}%
\AgdaSymbol{:}\AgdaSpace{}%
\AgdaPostulate{V}\AgdaSymbol{)}\AgdaSpace{}%
\AgdaSymbol{→}\AgdaSpace{}%
\AgdaPostulate{ze}\AgdaSpace{}%
\AgdaOperator{\AgdaDatatype{≡}}\AgdaSpace{}%
\AgdaSymbol{(}\AgdaOperator{\AgdaPostulate{\AgdaUnderscore{}+\AgdaUnderscore{}}}\AgdaSpace{}%
\AgdaSymbol{(}\AgdaOperator{\AgdaPostulate{-\AgdaUnderscore{}}}\AgdaSpace{}%
\AgdaBound{x0}\AgdaSymbol{)}\AgdaSpace{}%
\AgdaBound{x0}\AgdaSymbol{)}\<%
\\
\>[2]\AgdaFunction{vampagda-f9}\AgdaSpace{}%
\AgdaBound{x0}\AgdaSpace{}%
\AgdaSymbol{=}\AgdaSpace{}%
\AgdaSymbol{(}\AgdaFunction{sym}\AgdaSpace{}%
\AgdaSymbol{(}\AgdaFunction{vampagda-f2}\AgdaSpace{}%
\AgdaBound{x0}\AgdaSymbol{))}\<%
\\
\>[2]\AgdaFunction{vampagda-f10}\AgdaSpace{}%
\AgdaSymbol{:}\AgdaSpace{}%
\AgdaSymbol{∀}\AgdaSpace{}%
\AgdaSymbol{(}\AgdaBound{x0}\AgdaSpace{}%
\AgdaSymbol{:}\AgdaSpace{}%
\AgdaPostulate{V}\AgdaSymbol{)}\AgdaSpace{}%
\AgdaSymbol{(}\AgdaBound{x1}\AgdaSpace{}%
\AgdaSymbol{:}\AgdaSpace{}%
\AgdaPostulate{V}\AgdaSymbol{)}\AgdaSpace{}%
\AgdaSymbol{(}\AgdaBound{x2}\AgdaSpace{}%
\AgdaSymbol{:}\AgdaSpace{}%
\AgdaPostulate{V}\AgdaSymbol{)}\AgdaSpace{}%
\AgdaSymbol{→}\AgdaSpace{}%
\AgdaSymbol{(}\AgdaOperator{\AgdaPostulate{\AgdaUnderscore{}+\AgdaUnderscore{}}}\AgdaSpace{}%
\AgdaBound{x1}\AgdaSpace{}%
\AgdaSymbol{(}\AgdaOperator{\AgdaPostulate{\AgdaUnderscore{}+\AgdaUnderscore{}}}\AgdaSpace{}%
\AgdaBound{x2}\AgdaSpace{}%
\AgdaBound{x0}\AgdaSymbol{))}\AgdaSpace{}%
\AgdaOperator{\AgdaDatatype{≡}}\AgdaSpace{}%
\AgdaSymbol{(}\AgdaOperator{\AgdaPostulate{\AgdaUnderscore{}+\AgdaUnderscore{}}}\AgdaSpace{}%
\AgdaSymbol{(}\AgdaOperator{\AgdaPostulate{\AgdaUnderscore{}+\AgdaUnderscore{}}}\AgdaSpace{}%
\AgdaBound{x1}\AgdaSpace{}%
\AgdaBound{x2}\AgdaSymbol{)}\AgdaSpace{}%
\AgdaBound{x0}\AgdaSymbol{)}\<%
\\
\>[2]\AgdaFunction{vampagda-f10}\AgdaSpace{}%
\AgdaBound{x0}\AgdaSpace{}%
\AgdaBound{x1}\AgdaSpace{}%
\AgdaBound{x2}\AgdaSpace{}%
\AgdaSymbol{=}\AgdaSpace{}%
\AgdaSymbol{(}\AgdaFunction{sym}\AgdaSpace{}%
\AgdaSymbol{(}\AgdaFunction{vampagda-f3}\AgdaSpace{}%
\AgdaBound{x1}\AgdaSpace{}%
\AgdaBound{x2}\AgdaSpace{}%
\AgdaBound{x0}\AgdaSymbol{))}\<%
\\
\>[2]\AgdaFunction{vampagda-f11}\AgdaSpace{}%
\AgdaSymbol{:}\AgdaSpace{}%
\AgdaBound{v}\AgdaSpace{}%
\AgdaOperator{\AgdaDatatype{≡}}\AgdaSpace{}%
\AgdaSymbol{(}\AgdaOperator{\AgdaPostulate{\AgdaUnderscore{}+\AgdaUnderscore{}}}\AgdaSpace{}%
\AgdaBound{u}\AgdaSpace{}%
\AgdaBound{v}\AgdaSymbol{)}\<%
\\
\>[2]\AgdaFunction{vampagda-f11}\AgdaSpace{}%
\AgdaSymbol{=}\AgdaSpace{}%
\AgdaSymbol{(}\AgdaFunction{sym}\AgdaSpace{}%
\AgdaFunction{vampagda-f4}\AgdaSymbol{)}\<%
\\
\>[2]\AgdaFunction{vampagda-f12}\AgdaSpace{}%
\AgdaSymbol{:}\AgdaSpace{}%
\AgdaPostulate{ze}\AgdaSpace{}%
\AgdaOperator{\AgdaDatatype{≡}}\AgdaSpace{}%
\AgdaBound{u}\AgdaSpace{}%
\AgdaSymbol{→}\AgdaSpace{}%
\AgdaBound{u}\AgdaSpace{}%
\AgdaOperator{\AgdaDatatype{≡}}\AgdaSpace{}%
\AgdaPostulate{ze}\<%
\\
\>[2]\AgdaFunction{vampagda-f12}\AgdaSpace{}%
\AgdaBound{l0}\AgdaSpace{}%
\AgdaSymbol{=}\AgdaSpace{}%
\AgdaSymbol{(}\AgdaFunction{vampagda-f7}\AgdaSpace{}%
\AgdaBound{l0}\AgdaSymbol{)}\<%
\\
\>[2]\AgdaFunction{vampagda-f14}\AgdaSpace{}%
\AgdaSymbol{:}\AgdaSpace{}%
\AgdaSymbol{∀}\AgdaSpace{}%
\AgdaSymbol{(}\AgdaBound{x0}\AgdaSpace{}%
\AgdaSymbol{:}\AgdaSpace{}%
\AgdaPostulate{V}\AgdaSymbol{)}\AgdaSpace{}%
\AgdaSymbol{(}\AgdaBound{x1}\AgdaSpace{}%
\AgdaSymbol{:}\AgdaSpace{}%
\AgdaPostulate{V}\AgdaSymbol{)}\AgdaSpace{}%
\AgdaSymbol{→}\AgdaSpace{}%
\AgdaSymbol{(}\AgdaOperator{\AgdaPostulate{\AgdaUnderscore{}+\AgdaUnderscore{}}}\AgdaSpace{}%
\AgdaSymbol{(}\AgdaOperator{\AgdaPostulate{-\AgdaUnderscore{}}}\AgdaSpace{}%
\AgdaBound{x0}\AgdaSymbol{)}\AgdaSpace{}%
\AgdaSymbol{(}\AgdaOperator{\AgdaPostulate{\AgdaUnderscore{}+\AgdaUnderscore{}}}\AgdaSpace{}%
\AgdaBound{x0}\AgdaSpace{}%
\AgdaBound{x1}\AgdaSymbol{))}\AgdaSpace{}%
\AgdaOperator{\AgdaDatatype{≡}}\AgdaSpace{}%
\AgdaSymbol{(}\AgdaOperator{\AgdaPostulate{\AgdaUnderscore{}+\AgdaUnderscore{}}}\AgdaSpace{}%
\AgdaPostulate{ze}\AgdaSpace{}%
\AgdaBound{x1}\AgdaSymbol{)}\<%
\\
\>[2]\AgdaFunction{vampagda-f14}\AgdaSpace{}%
\AgdaBound{x0}\AgdaSpace{}%
\AgdaBound{x1}\AgdaSpace{}%
\AgdaSymbol{=}\AgdaSpace{}%
\AgdaSymbol{(}\AgdaFunction{trans}\AgdaSpace{}%
\AgdaSymbol{(}\AgdaFunction{vampagda-f10}\AgdaSpace{}%
\AgdaBound{x1}\AgdaSpace{}%
\AgdaSymbol{(}\AgdaOperator{\AgdaPostulate{-\AgdaUnderscore{}}}\AgdaSpace{}%
\AgdaBound{x0}\AgdaSymbol{)}\AgdaSpace{}%
\AgdaBound{x0}\AgdaSymbol{)}\AgdaSpace{}%
\AgdaSymbol{(}\AgdaFunction{cong}\AgdaSpace{}%
\AgdaSymbol{(λ}\AgdaSpace{}%
\AgdaBound{❌}\AgdaSpace{}%
\AgdaSymbol{→}\AgdaSpace{}%
\AgdaSymbol{(}\AgdaOperator{\AgdaPostulate{\AgdaUnderscore{}+\AgdaUnderscore{}}}\AgdaSpace{}%
\AgdaBound{❌}\AgdaSpace{}%
\AgdaBound{x1}\AgdaSymbol{))}\AgdaSpace{}%
\AgdaSymbol{(}\AgdaFunction{sym}\AgdaSpace{}%
\AgdaSymbol{(}\AgdaFunction{vampagda-f9}\AgdaSpace{}%
\AgdaBound{x0}\AgdaSymbol{))))}\<%
\\
\>[2]\AgdaFunction{vampagda-f16}\AgdaSpace{}%
\AgdaSymbol{:}\AgdaSpace{}%
\AgdaSymbol{∀}\AgdaSpace{}%
\AgdaSymbol{(}\AgdaBound{x0}\AgdaSpace{}%
\AgdaSymbol{:}\AgdaSpace{}%
\AgdaPostulate{V}\AgdaSymbol{)}\AgdaSpace{}%
\AgdaSymbol{→}\AgdaSpace{}%
\AgdaSymbol{(}\AgdaOperator{\AgdaPostulate{\AgdaUnderscore{}+\AgdaUnderscore{}}}\AgdaSpace{}%
\AgdaBound{v}\AgdaSpace{}%
\AgdaBound{x0}\AgdaSymbol{)}\AgdaSpace{}%
\AgdaOperator{\AgdaDatatype{≡}}\AgdaSpace{}%
\AgdaSymbol{(}\AgdaOperator{\AgdaPostulate{\AgdaUnderscore{}+\AgdaUnderscore{}}}\AgdaSpace{}%
\AgdaBound{u}\AgdaSpace{}%
\AgdaSymbol{(}\AgdaOperator{\AgdaPostulate{\AgdaUnderscore{}+\AgdaUnderscore{}}}\AgdaSpace{}%
\AgdaBound{v}\AgdaSpace{}%
\AgdaBound{x0}\AgdaSymbol{))}\<%
\\
\>[2]\AgdaFunction{vampagda-f16}\AgdaSpace{}%
\AgdaBound{x0}\AgdaSpace{}%
\AgdaSymbol{=}\AgdaSpace{}%
\AgdaSymbol{(}\AgdaFunction{trans}\AgdaSpace{}%
\AgdaSymbol{(}\AgdaFunction{cong}\AgdaSpace{}%
\AgdaSymbol{(λ}\AgdaSpace{}%
\AgdaBound{❌}\AgdaSpace{}%
\AgdaSymbol{→}\AgdaSpace{}%
\AgdaSymbol{(}\AgdaOperator{\AgdaPostulate{\AgdaUnderscore{}+\AgdaUnderscore{}}}\AgdaSpace{}%
\AgdaBound{❌}\AgdaSpace{}%
\AgdaBound{x0}\AgdaSymbol{))}\AgdaSpace{}%
\AgdaFunction{vampagda-f11}\AgdaSymbol{)}\AgdaSpace{}%
\AgdaSymbol{(}\AgdaFunction{sym}\AgdaSpace{}%
\AgdaSymbol{(}\AgdaFunction{vampagda-f10}\AgdaSpace{}%
\AgdaBound{x0}\AgdaSpace{}%
\AgdaBound{u}\AgdaSpace{}%
\AgdaBound{v}\AgdaSymbol{)))}\<%
\\
\>[2]\AgdaFunction{vampagda-f17}\AgdaSpace{}%
\AgdaSymbol{:}\AgdaSpace{}%
\AgdaSymbol{∀}\AgdaSpace{}%
\AgdaSymbol{(}\AgdaBound{x0}\AgdaSpace{}%
\AgdaSymbol{:}\AgdaSpace{}%
\AgdaPostulate{V}\AgdaSymbol{)}\AgdaSpace{}%
\AgdaSymbol{(}\AgdaBound{x1}\AgdaSpace{}%
\AgdaSymbol{:}\AgdaSpace{}%
\AgdaPostulate{V}\AgdaSymbol{)}\AgdaSpace{}%
\AgdaSymbol{→}\AgdaSpace{}%
\AgdaSymbol{(}\AgdaOperator{\AgdaPostulate{\AgdaUnderscore{}+\AgdaUnderscore{}}}\AgdaSpace{}%
\AgdaSymbol{(}\AgdaOperator{\AgdaPostulate{-\AgdaUnderscore{}}}\AgdaSpace{}%
\AgdaBound{x0}\AgdaSymbol{)}\AgdaSpace{}%
\AgdaSymbol{(}\AgdaOperator{\AgdaPostulate{\AgdaUnderscore{}+\AgdaUnderscore{}}}\AgdaSpace{}%
\AgdaBound{x0}\AgdaSpace{}%
\AgdaBound{x1}\AgdaSymbol{))}\AgdaSpace{}%
\AgdaOperator{\AgdaDatatype{≡}}\AgdaSpace{}%
\AgdaBound{x1}\<%
\\
\>[2]\AgdaFunction{vampagda-f17}\AgdaSpace{}%
\AgdaBound{x0}\AgdaSpace{}%
\AgdaBound{x1}\AgdaSpace{}%
\AgdaSymbol{=}\AgdaSpace{}%
\AgdaSymbol{(}\AgdaFunction{trans}\AgdaSpace{}%
\AgdaSymbol{(}\AgdaFunction{vampagda-f14}\AgdaSpace{}%
\AgdaBound{x0}\AgdaSpace{}%
\AgdaBound{x1}\AgdaSymbol{)}\AgdaSpace{}%
\AgdaSymbol{(}\AgdaFunction{vampagda-f8}\AgdaSpace{}%
\AgdaBound{x1}\AgdaSymbol{))}\<%
\\
\>[2]\AgdaFunction{vampagda-f21}\AgdaSpace{}%
\AgdaSymbol{:}\AgdaSpace{}%
\AgdaSymbol{∀}\AgdaSpace{}%
\AgdaSymbol{(}\AgdaBound{x0}\AgdaSpace{}%
\AgdaSymbol{:}\AgdaSpace{}%
\AgdaPostulate{V}\AgdaSymbol{)}\AgdaSpace{}%
\AgdaSymbol{→}\AgdaSpace{}%
\AgdaSymbol{(}\AgdaOperator{\AgdaPostulate{\AgdaUnderscore{}+\AgdaUnderscore{}}}\AgdaSpace{}%
\AgdaSymbol{(}\AgdaOperator{\AgdaPostulate{-\AgdaUnderscore{}}}\AgdaSpace{}%
\AgdaSymbol{(}\AgdaOperator{\AgdaPostulate{-\AgdaUnderscore{}}}\AgdaSpace{}%
\AgdaBound{x0}\AgdaSymbol{))}\AgdaSpace{}%
\AgdaPostulate{ze}\AgdaSymbol{)}\AgdaSpace{}%
\AgdaOperator{\AgdaDatatype{≡}}\AgdaSpace{}%
\AgdaBound{x0}\<%
\\
\>[2]\AgdaFunction{vampagda-f21}\AgdaSpace{}%
\AgdaBound{x0}\AgdaSpace{}%
\AgdaSymbol{=}\AgdaSpace{}%
\AgdaSymbol{(}\AgdaFunction{trans}\AgdaSpace{}%
\AgdaSymbol{(}\AgdaFunction{cong}\AgdaSpace{}%
\AgdaSymbol{(λ}\AgdaSpace{}%
\AgdaBound{❌}\AgdaSpace{}%
\AgdaSymbol{→}\AgdaSpace{}%
\AgdaSymbol{(}\AgdaOperator{\AgdaPostulate{\AgdaUnderscore{}+\AgdaUnderscore{}}}\AgdaSpace{}%
\AgdaSymbol{(}\AgdaOperator{\AgdaPostulate{-\AgdaUnderscore{}}}\AgdaSpace{}%
\AgdaSymbol{(}\AgdaOperator{\AgdaPostulate{-\AgdaUnderscore{}}}\AgdaSpace{}%
\AgdaBound{x0}\AgdaSymbol{))}\AgdaSpace{}%
\AgdaBound{❌}\AgdaSymbol{))}\AgdaSpace{}%
\AgdaSymbol{(}\AgdaFunction{vampagda-f9}\AgdaSpace{}%
\AgdaBound{x0}\AgdaSymbol{))}\AgdaSpace{}%
\AgdaSymbol{(}\AgdaFunction{vampagda-f17}\AgdaSpace{}%
\AgdaSymbol{(}\AgdaOperator{\AgdaPostulate{-\AgdaUnderscore{}}}\AgdaSpace{}%
\AgdaBound{x0}\AgdaSymbol{)}\AgdaSpace{}%
\AgdaBound{x0}\AgdaSymbol{))}\<%
\\
\>[2]\AgdaFunction{vampagda-f22}\AgdaSpace{}%
\AgdaSymbol{:}\AgdaSpace{}%
\AgdaSymbol{∀}\AgdaSpace{}%
\AgdaSymbol{(}\AgdaBound{x0}\AgdaSpace{}%
\AgdaSymbol{:}\AgdaSpace{}%
\AgdaPostulate{V}\AgdaSymbol{)}\AgdaSpace{}%
\AgdaSymbol{(}\AgdaBound{x1}\AgdaSpace{}%
\AgdaSymbol{:}\AgdaSpace{}%
\AgdaPostulate{V}\AgdaSymbol{)}\AgdaSpace{}%
\AgdaSymbol{→}\AgdaSpace{}%
\AgdaSymbol{(}\AgdaOperator{\AgdaPostulate{\AgdaUnderscore{}+\AgdaUnderscore{}}}\AgdaSpace{}%
\AgdaBound{x1}\AgdaSpace{}%
\AgdaBound{x0}\AgdaSymbol{)}\AgdaSpace{}%
\AgdaOperator{\AgdaDatatype{≡}}\AgdaSpace{}%
\AgdaSymbol{(}\AgdaOperator{\AgdaPostulate{\AgdaUnderscore{}+\AgdaUnderscore{}}}\AgdaSpace{}%
\AgdaSymbol{(}\AgdaOperator{\AgdaPostulate{-\AgdaUnderscore{}}}\AgdaSpace{}%
\AgdaSymbol{(}\AgdaOperator{\AgdaPostulate{-\AgdaUnderscore{}}}\AgdaSpace{}%
\AgdaBound{x1}\AgdaSymbol{))}\AgdaSpace{}%
\AgdaBound{x0}\AgdaSymbol{)}\<%
\\
\>[2]\AgdaFunction{vampagda-f22}\AgdaSpace{}%
\AgdaBound{x0}\AgdaSpace{}%
\AgdaBound{x1}\AgdaSpace{}%
\AgdaSymbol{=}\AgdaSpace{}%
\AgdaSymbol{(}\AgdaFunction{trans}\AgdaSpace{}%
\AgdaSymbol{(}\AgdaFunction{sym}\AgdaSpace{}%
\AgdaSymbol{(}\AgdaFunction{vampagda-f17}\AgdaSpace{}%
\AgdaSymbol{(}\AgdaOperator{\AgdaPostulate{-\AgdaUnderscore{}}}\AgdaSpace{}%
\AgdaBound{x1}\AgdaSymbol{)}\AgdaSpace{}%
\AgdaSymbol{(}\AgdaOperator{\AgdaPostulate{\AgdaUnderscore{}+\AgdaUnderscore{}}}\AgdaSpace{}%
\AgdaBound{x1}\AgdaSpace{}%
\AgdaBound{x0}\AgdaSymbol{)))}\AgdaSpace{}%
\AgdaSymbol{(}\AgdaFunction{cong}\AgdaSpace{}%
\AgdaSymbol{(λ}\AgdaSpace{}%
\AgdaBound{❌}\AgdaSpace{}%
\AgdaSymbol{→}\AgdaSpace{}%
\AgdaSymbol{(}\AgdaOperator{\AgdaPostulate{\AgdaUnderscore{}+\AgdaUnderscore{}}}\AgdaSpace{}%
\AgdaSymbol{(}\AgdaOperator{\AgdaPostulate{-\AgdaUnderscore{}}}\AgdaSpace{}%
\AgdaSymbol{(}\AgdaOperator{\AgdaPostulate{-\AgdaUnderscore{}}}\AgdaSpace{}%
\AgdaBound{x1}\AgdaSymbol{))}\AgdaSpace{}%
\AgdaBound{❌}\AgdaSymbol{))}\AgdaSpace{}%
\AgdaSymbol{(}\AgdaFunction{vampagda-f17}\AgdaSpace{}%
\AgdaBound{x1}\AgdaSpace{}%
\AgdaBound{x0}\AgdaSymbol{)))}\<%
\\
\>[2]\AgdaFunction{vampagda-f27}\AgdaSpace{}%
\AgdaSymbol{:}\AgdaSpace{}%
\AgdaSymbol{∀}\AgdaSpace{}%
\AgdaSymbol{(}\AgdaBound{x0}\AgdaSpace{}%
\AgdaSymbol{:}\AgdaSpace{}%
\AgdaPostulate{V}\AgdaSymbol{)}\AgdaSpace{}%
\AgdaSymbol{→}\AgdaSpace{}%
\AgdaSymbol{(}\AgdaOperator{\AgdaPostulate{\AgdaUnderscore{}+\AgdaUnderscore{}}}\AgdaSpace{}%
\AgdaBound{x0}\AgdaSpace{}%
\AgdaPostulate{ze}\AgdaSymbol{)}\AgdaSpace{}%
\AgdaOperator{\AgdaDatatype{≡}}\AgdaSpace{}%
\AgdaBound{x0}\<%
\\
\>[2]\AgdaFunction{vampagda-f27}\AgdaSpace{}%
\AgdaBound{x0}\AgdaSpace{}%
\AgdaSymbol{=}\AgdaSpace{}%
\AgdaSymbol{(}\AgdaFunction{trans}\AgdaSpace{}%
\AgdaSymbol{(}\AgdaFunction{vampagda-f22}\AgdaSpace{}%
\AgdaPostulate{ze}\AgdaSpace{}%
\AgdaBound{x0}\AgdaSymbol{)}\AgdaSpace{}%
\AgdaSymbol{(}\AgdaFunction{vampagda-f21}\AgdaSpace{}%
\AgdaBound{x0}\AgdaSymbol{))}\<%
\\
\>[2]\AgdaFunction{vampagda-f56}\AgdaSpace{}%
\AgdaSymbol{:}\AgdaSpace{}%
\AgdaSymbol{∀}\AgdaSpace{}%
\AgdaSymbol{(}\AgdaBound{x0}\AgdaSpace{}%
\AgdaSymbol{:}\AgdaSpace{}%
\AgdaPostulate{V}\AgdaSymbol{)}\AgdaSpace{}%
\AgdaSymbol{→}\AgdaSpace{}%
\AgdaPostulate{ze}\AgdaSpace{}%
\AgdaOperator{\AgdaDatatype{≡}}\AgdaSpace{}%
\AgdaSymbol{(}\AgdaOperator{\AgdaPostulate{\AgdaUnderscore{}+\AgdaUnderscore{}}}\AgdaSpace{}%
\AgdaBound{x0}\AgdaSpace{}%
\AgdaSymbol{(}\AgdaOperator{\AgdaPostulate{-\AgdaUnderscore{}}}\AgdaSpace{}%
\AgdaBound{x0}\AgdaSymbol{))}\<%
\\
\>[2]\AgdaFunction{vampagda-f56}\AgdaSpace{}%
\AgdaBound{x0}\AgdaSpace{}%
\AgdaSymbol{=}\AgdaSpace{}%
\AgdaSymbol{(}\AgdaFunction{trans}\AgdaSpace{}%
\AgdaSymbol{(}\AgdaFunction{vampagda-f9}\AgdaSpace{}%
\AgdaSymbol{(}\AgdaOperator{\AgdaPostulate{-\AgdaUnderscore{}}}\AgdaSpace{}%
\AgdaBound{x0}\AgdaSymbol{))}\AgdaSpace{}%
\AgdaSymbol{(}\AgdaFunction{sym}\AgdaSpace{}%
\AgdaSymbol{(}\AgdaFunction{vampagda-f22}\AgdaSpace{}%
\AgdaSymbol{(}\AgdaOperator{\AgdaPostulate{-\AgdaUnderscore{}}}\AgdaSpace{}%
\AgdaBound{x0}\AgdaSymbol{)}\AgdaSpace{}%
\AgdaBound{x0}\AgdaSymbol{)))}\<%
\\
\>[2]\AgdaFunction{vampagda-f87}\AgdaSpace{}%
\AgdaSymbol{:}\AgdaSpace{}%
\AgdaPostulate{ze}\AgdaSpace{}%
\AgdaOperator{\AgdaDatatype{≡}}\AgdaSpace{}%
\AgdaSymbol{(}\AgdaOperator{\AgdaPostulate{\AgdaUnderscore{}+\AgdaUnderscore{}}}\AgdaSpace{}%
\AgdaBound{u}\AgdaSpace{}%
\AgdaPostulate{ze}\AgdaSymbol{)}\<%
\\
\>[2]\AgdaFunction{vampagda-f87}\AgdaSpace{}%
\AgdaSymbol{=}\AgdaSpace{}%
\AgdaSymbol{(}\AgdaFunction{trans}\AgdaSpace{}%
\AgdaSymbol{(}\AgdaFunction{trans}\AgdaSpace{}%
\AgdaSymbol{(}\AgdaFunction{cong}\AgdaSpace{}%
\AgdaSymbol{(λ}\AgdaSpace{}%
\AgdaBound{❌}\AgdaSpace{}%
\AgdaSymbol{→}\AgdaSpace{}%
\AgdaBound{❌}\AgdaSymbol{)}\AgdaSpace{}%
\AgdaSymbol{(}\AgdaFunction{vampagda-f56}\AgdaSpace{}%
\AgdaBound{v}\AgdaSymbol{))}\AgdaSpace{}%
\AgdaSymbol{(}\AgdaFunction{vampagda-f16}\AgdaSpace{}%
\AgdaSymbol{(}\AgdaOperator{\AgdaPostulate{-\AgdaUnderscore{}}}\AgdaSpace{}%
\AgdaBound{v}\AgdaSymbol{)))}\AgdaSpace{}%
\AgdaSymbol{(}\AgdaFunction{cong}\AgdaSpace{}%
\AgdaSymbol{(λ}\AgdaSpace{}%
\AgdaBound{❌}\AgdaSpace{}%
\AgdaSymbol{→}\AgdaSpace{}%
\AgdaSymbol{(}\AgdaOperator{\AgdaPostulate{\AgdaUnderscore{}+\AgdaUnderscore{}}}\AgdaSpace{}%
\AgdaBound{u}\AgdaSpace{}%
\AgdaBound{❌}\AgdaSymbol{))}\AgdaSpace{}%
\AgdaSymbol{(}\AgdaFunction{sym}\AgdaSpace{}%
\AgdaSymbol{(}\AgdaFunction{vampagda-f56}\AgdaSpace{}%
\AgdaBound{v}\AgdaSymbol{))))}\<%
\\
\>[2]\AgdaFunction{vampagda-f88}\AgdaSpace{}%
\AgdaSymbol{:}\AgdaSpace{}%
\AgdaPostulate{ze}\AgdaSpace{}%
\AgdaOperator{\AgdaDatatype{≡}}\AgdaSpace{}%
\AgdaBound{u}\<%
\\
\>[2]\AgdaFunction{vampagda-f88}\AgdaSpace{}%
\AgdaSymbol{=}\AgdaSpace{}%
\AgdaSymbol{(}\AgdaFunction{trans}\AgdaSpace{}%
\AgdaFunction{vampagda-f87}\AgdaSpace{}%
\AgdaSymbol{(}\AgdaFunction{vampagda-f27}\AgdaSpace{}%
\AgdaBound{u}\AgdaSymbol{))}\<%
\\
\>[2]\AgdaFunction{vampagda-f93}\AgdaSpace{}%
\AgdaSymbol{:}\AgdaSpace{}%
\AgdaBound{u}\AgdaSpace{}%
\AgdaOperator{\AgdaDatatype{≡}}\AgdaSpace{}%
\AgdaPostulate{ze}\<%
\\
\>[2]\AgdaFunction{vampagda-f93}\AgdaSpace{}%
\AgdaSymbol{=}\AgdaSpace{}%
\AgdaSymbol{(}\AgdaFunction{vampagda-f12}\AgdaSpace{}%
\AgdaFunction{vampagda-f88}\AgdaSymbol{)}\<%
\end{code}

\subsection{Invoking \vampire}
We now have an SMT-LIB input that \vampire can attempt.
\vampire, like most ATPs, works best if allowed to run a \emph{portfolio}, trying many different \emph{strategies} until one succeeds.
A strategy is simply a collection of options for \vampire.
The idea is that a single strategy is likely to find a proof quickly or not at all, so it is better to restart afresh regularly with a new strategy than continue to try hard for a long time.
We found that the pre-baked CASC~\cite{casc} portfolio was highly effective on the unseen\footnote{this lends credence to the idea that the portfolio was adequately regularised during construction~\cite{spider-regularisation}} Agda problems.
Portfolio modes also allow making full use of available parallelism by running multiple strategies simultaneously.

However, running the CASC portfolio also exposes us to the full range of possible \vampire inferences, some of which can be very surprising.
We have not implemented \emph{all} of \vampire's inferences for this prototype, but we have done enough such that all the proofs in \cref{sec:running} are found and reconstructed while using the portfolio.
\vampire delivers the following proof for the SMT-LIB input above:
\begin{enumerate}
	\item $\forall x.~0 + x = x$ (\texttt{neutl})
	\item $\forall x.~0 = -x + x$ (\texttt{negl})
	\item $\forall xyz.~(x + y) + z = x + (y + z)$ (\texttt{assoc})
	\item $v = u + v$ (\texttt{e})
	\item $0 \neq u$ (negated goal)
	\item $\forall xy.~-x + (x + y) = 0 + y$ (superposition 3, 2)
	\item $\forall x.~v + x = u + (v + x) $ (superposition 3, 4)
	\item $\forall xy.~-x + (x + y) = y$ (superposition 6, 1)
	\item $\forall x.~(--x) + 0 = x$ (superposition 8, 2)
	\item $\forall xy.~y + x = (--y) + x$ (superposition 8, 8)
	\item $\forall x. x + 0 = x$ (superposition 9, 10)
	\item $\forall x.~0 = x + -x$ (superposition 2, 10)
	\item $0 = u + 0$ (superposition 7, 12)
	\item $0 = u$ (superposition 13, 11)
	\item $\bot$ (resolution 14, 5)
\end{enumerate}
Note that the proof ends in $\bot$, and requires the negated goal as a premise.
\vampire treats equations as symmetric, so superpositions may occur in either direction.

\subsection{Transforming the Proof}
\label{sec:proof-transform}
At this stage we we have a set of definite clauses $\Gamma$, a ground atom $G$, and a classical proof from \vampire purporting to show $\Gamma \vdash G$.
\vampire will in fact show $\Gamma, \lnot G \vdash \bot$.
The input set is already in CNF, so \vampire's preprocessing rules do not apply and the proof of $\Gamma, \lnot G \vdash \bot$ will only involve rules of the superposition calculus given above\footnote{this is not quite true, but the details are not important here: more on this in \cref{sec:edge-cases}}.
Let us specialise \vampire's superposition calculus to operate on Horn premises:
\begin{mathpar}
\inferrule[Resolution]{C \to P \\ P' \wedge D \to \boxed{Q}}{\sigma(C \wedge D \to \boxed{Q})}\and
\inferrule[Factoring]{P \wedge P' \wedge C \to \boxed{Q}}{\sigma(P \wedge C \to \boxed{Q})}\and
\inferrule[Equality Resolution]{s = t \wedge C \to \boxed{P}}{\sigma(C \to \boxed{P})}\and
\inferrule[Superposition Left]{C \to l = r \\ P[l'] \wedge D \to \boxed{Q}}{\sigma(C \wedge P[r] \wedge D \to \boxed{Q})}\and
\inferrule[Superposition Right]{C \to l = r \\ D \to P[l']}{\sigma(C \wedge D \to P[r])}
\end{mathpar}
where $\sigma = \mgu(P, P')$ in \textsc{Resolution} and \textsc{Factoring}, $\mgu(l, l')$ in both \textsc{Superposition} variants, and $\mgu(s, t)$ in \textsc{Equality Resolution}.
Each conclusion is Horn.
It is not too hard to see that these rules are constructively valid, but we will show this definitively by giving a translation to Agda.
Now consider the \fbox{boxed} atoms in the premises of some rules: these atoms could be $\bot$, i.e. when the premise is a goal clause.
When the boxed atom is \emph{ground} (including, but not limited to, $\bot$), each such rule simply duplicates the boxed atom in the conclusion.
This is because substitution on ground atoms does nothing.
\begin{lemma}[agnostic]
\label{lemma:agnostic}
For any inference of the Horn superposition calculus with a premise of the form $C \to \bot$, the conclusion is of the form $D \to \bot$.
Furthermore, a difference instance of the same rule concludes $D \to G$ from $C \to G$ for any ground atom $G$.
\end{lemma}
\begin{proof}
By case analysis on the rules of the calculus.
\end{proof}
\cref{lemma:agnostic} permits a very simple transformation of the proof.
First, replace the root goal clause $G \to \bot$ with the tautology $G \to G$.
Then, for any inference in the proof that has a goal clause $C \to \bot$ as a premise, replace it with $C \to G$.
This will produce a conclusion of the form $D \to G$.
The final empty goal clause $\to \bot$ in the refutation will be transformed into a clause $\to G$, so we have $\Gamma, (G \to G) \vdash G$, from which we obtain $\Gamma \vdash G$.
In the example above, step 5 is translated to $0 = u \to 0 = u$ and step 15 becomes $0 = u$.
This is a simple case, but in general $\bot$ may be woven through the entire proof.

\subsubsection{On Friedman}
\label{sec:friedman}
This trick is similar to the technique used by Friedman~\cite{friedman} to show that any $\Pi^0_2$ sentence that is a theorem of Peano arithmetic is also a theorem of Heyting arithmetic.
We follow Selinger's presentation~\cite{selinger}.
Friedman defines the $A$-translation of a formula $F$ to be the result of replacing any atomic formula $B$ within $F$ (including $\bot$) with $B \vee A$, provided no free variable of $A$ is bound in $F$: this is shown to preserve validity in Heyting arithmetic.
$\Pi^0_2$ sentences can be reduced to the form $\exists y.~f(x, y) = 0$ for some primitive recursive $f$: let us call this $\phi$.
The $A$-translation of $\phi$ is $\exists y.~(f(x, y) = 0 \vee A)$, which is equivalent to $\phi \vee A$.
Friedman first shows that if $\pa \phi$, then $\ha (\phi \to \bot) \to \bot$.
Then, with $A = \phi$, $A$-translation yields
$$\ha ((\phi \vee \phi) \to \phi) \to \phi$$
and hence $\ha \phi$.
Selinger states that as long as a theory can prove its own axioms in both double-negated and $A$-translated forms, Friedman's argument will work.
Our theory is not Heyting arithmetic, but the user-provided axioms.
However, since for any Horn premise we can prove the double-negated form and the $A$-translation, we know that a proof of $(G \to \bot) \to \bot$ can be transformed into a proof of $G$.
The approach we employ is not precisely the same as Friedman's, but we could certainly say it rhymes.

\subsection{Translating the Proof to Agda}
\label{sec:vampire-to-agda}
We now need to translate any inference of the above Horn superposition calculus into an Agda proof term.
This is neither difficult nor novel: a very similar translation is given in a slightly different context by Burel~\cite{burel},
although we do not have a doubly-negated representation.
However, we sketch the scheme here in order to give the main idea.
For the sake of clarity we present a ground calculus and leave it to the reader to insert the appropriate introduction and instantiation of variables.
Corresponding with the Horn inference rules above:
\begin{mathpar}
\inferrule[Resolution]
{t_1 : C \to P \\ t_2 : P \to D \to Q}
{\left(\lambda cd.~t_2~\left(t_1~c\right)~d\right) : C \to D \to Q}\and
\inferrule[Factoring]{t : P \to P \to C \to Q}
{\left(\lambda pc.~t~p~p~c\right) : P \to C \to Q}\and
\inferrule[Equality Factoring]{t : s = s \to C \to P}
{\left(\lambda c.~t~\mathrm{refl}~c\right) : C \to P}\and
\inferrule[Superposition Right]
{t_1 : C \to l = r \\ t_2 : D \to P[l]}
{\left(\lambda cd.~\mathrm{rw}~\left(t_1~c\right)~\left(t_2~d\right)\right): C \to D \to P[r]}\and
\inferrule[Superposition Left]
{t_1 : C \to l = r \\ t_2 : P[l] \to D \to Q}
{\left(\lambda cpd.~t_2~\left(\mathrm{rw}~(\mathrm{sym}~(t_1~c))~p\right)~d\right): C \to P[r] \to D \to Q}
\end{mathpar}
where $\mathrm{rw} : l = r \to P[l] \to P[r]$ is an appropriate Agda ``rewrite'' term for any $P$, $l$ and $r$, and $\mathrm{sym} : l = r \to r = l$ is a lemma showing symmetry of equality.

\subsubsection{Implementation via Prolog}
Producing precise proof terms from ATP proofs is usually done by either modifying the ATP to output them directly~\cite{vampukti}, or by invoking a different, certificate-producing, ATP on each proof step, producing a certificate for that step~\cite{gdv}.
Neither option was attractive for us.
Modifying \vampire is a great deal of work, causing a (small but present) performance hit and producing a maintenance burden.
Invoking a certifying ATP on proof steps is also considerable integration work, and occasionally the certifying ATP times out in hard-to-predict ways.

Instead, we exploit the fact that \vampire produces TSTP~\cite{tptp} proofs, which are valid Prolog~\cite{iso-prolog} with some special operators defined.
We implement the TSTP-Agda translation with a Prolog script.
The script first reads the proof into an internal database using \texttt{read/1}, recording the name (e.g. \texttt{superposition}) of inferences and their premises.
At this stage, we ensure that clauses are Horn and replace any $\bot$ with $G$, as above.
Then, for each step in the proof, the script performs backtracking \emph{search} for a valid e.g. superposition step that produces the correct conclusion given the premises, building an Agda proof term as it goes.
The script handles \vampire's idiosyncrasies, including variable renaming and implicitly treating equality as symmetric.
On success, the proof term for the step is complete.

\textbf{We highly recommend this method for similar tasks}.
Many aspects which are painful with other methods, chiefly variable renaming, are non-problems here as Prolog already has the correct semantics.
Backtracking search for valid inferences can be implemented concisely and declaratively: the entire Prolog script, including loading, reconstruction, printing, and several extensions, is achieved in less than 500 lines of inexpert code.
In principle this search could result in combinatorial blowup, but this does not seem to be a problem in practice and was not seen during testing on the entire supported fragment of the TPTP~\cite{tptp} benchmark set.
We think this improved behaviour is because the Prolog script \emph{knows it is looking for a single inference}, whereas an external ATP does not.

\subsubsection{Edge Cases and Other Inferences}
\label{sec:edge-cases}
There is one final mismatch between \vampire's logic and Agda's: empty domains.
In classical first-order logic, all domains are considered implicitly non-empty, so that $\forall x.~P(x)$ and $\forall x.~(P(x) \implies Q)$ entail $Q$.
However, this does not always hold in Agda as the domain of $x$ may have \emph{no} inhabitants.
When this situation occurs, the Prolog script looks around for bound variables or constants of the appropriate sort to show that the domain is inhabited.
In the rare event that this still fails, a hole is left for the user to fill.

\vampire also has more inferences than we admitted earlier: at the time of writing it has around 200 (!).
This was not a concern in practice: many relate to inputs like arithmetic outside of the fragment we consider here; some relate to translation into clause normal form, unnecessary here; and some are special cases of the core calculus and can be dealt with as such.
We must only handle the remainder, which turn out to be straightforward.
\vampire likes to introduce and eliminate definitions for terms as it sees fit, and we must go along with it.
Introducing definitions is easy, for they simply become Agda definitions.
Eliminating definitions is more involved but is essentially equational reasoning of the same flavour as the superposition rule.
\vampire also knows about the properties of inductively-defined datatypes, and will e.g. deduce $x = y$ and $l = k$ from $x :: l = y :: k$.
This kind of reasoning turns out to be mechanical and easy to translate.
In the end, we handle 21 named inferences.

Some inferences will prove difficult to handle fully in this context, including those which escape the Horn fragment, introduce new symbols without an obvious definition, or perform a reduction to SAT~\cite{uses-of-sat}.
If desired, it is possible to disable \vampire's more flamboyant options in exchange for reduced proving power: we would recommend the command line flag \verb|--forced_options| \verb|av=off:ins=0:gsp=off:gs=off:tha=off:ep=off:gtg=off|.

\begin{code}[hide]%
\>[0]\AgdaComment{--module\ papers.ITP26.running-example\ where}\<%
\\
\\%
\>[0]\AgdaKeyword{open}\AgdaSpace{}%
\AgdaKeyword{import}\AgdaSpace{}%
\AgdaModule{Relation.Binary.PropositionalEquality}\<%
\\
\>[0]\AgdaKeyword{open}\AgdaSpace{}%
\AgdaKeyword{import}\AgdaSpace{}%
\AgdaModule{Data.Nat}\AgdaSpace{}%
\AgdaSymbol{as}\AgdaSpace{}%
\AgdaModule{ℕ}\AgdaSpace{}%
\AgdaKeyword{using}\AgdaSpace{}%
\AgdaSymbol{(}\AgdaDatatype{ℕ}\AgdaSymbol{;}\AgdaSpace{}%
\AgdaInductiveConstructor{zero}\AgdaSymbol{;}\AgdaSpace{}%
\AgdaInductiveConstructor{suc}\AgdaSymbol{)}\<%
\\
\>[0]\AgdaKeyword{import}\AgdaSpace{}%
\AgdaModule{Algebra.Structures}\AgdaSpace{}%
\AgdaSymbol{as}\AgdaSpace{}%
\AgdaModule{AlgebraStructures}\<%
\\
\>[0]\AgdaKeyword{import}\AgdaSpace{}%
\AgdaModule{Algebra.Definitions}\AgdaSpace{}%
\AgdaSymbol{as}\AgdaSpace{}%
\AgdaModule{AlgebraDefinitions}\<%
\\
\>[0]\AgdaKeyword{module}\AgdaSpace{}%
\AgdaModule{\AgdaUnderscore{}}\AgdaSpace{}%
\AgdaKeyword{where}\<%
\end{code}

\section{\label{sec:running}Case Study}

To illustrate the practicality of the proposed Agda–Vampire integration 
we present a case study drawn from a real‑world formalisation effort.
The problem was selected because of:
\begin{description}
  \item[practical relevance] --- it arose during one of our own formalisation
    projects;
  \item[non‑triviality] --- the built‑in proof search was unable to discharge
    neither of the goal, and a seasoned Agda developer required roughly two days
    to complete the proof (available in supplementary materials as
    \texttt{ICplx.agda});
  \item[structure amenable to automation] --- the difficulty lies primarily in
    mechanically assembling the constituent lemmas, whereas the underlying
    mathematical insight is comparatively straightforward.  Consequently, the
    primary aim is not to discover novel proof structure but to automate the
    tedious goal‑discharge phase.
\end{description}
When formalising Fast Fourier Transform (FFT) algorithms and proving their
correctness, we wish to abstract away from the classical presentation that
relies on complex numbers.  Instead we aim to generalise the algorithm to any
commutative ring equipped with roots of unity.  As a sanity check we must
verify that our chosen axioms behave correctly for the familiar case of complex
numbers and their roots of unity. To this end we first introduce an axiomatic
theory of the reals together with elementary trigonometry.  From these
foundations we derive the conventional representation of complex numbers as
ordered pairs of reals and define roots of unity via complex exponentiation.
Finally we employ our system to confirm that this encoding satisfies the
required properties.

We begin by declaring the interface that captures the algebraic operations and
their associated laws.  The interface consists of a carrier set \AF{F}
together with addition, negation, multiplication, distinguished constants
\AF{𝟘ᶠ} and \AF{𝟙ᶠ}, and a primitive for roots of unity, \AF{-ωᶠ}.  The latter
denotes the primitive $n$-th root of
unity raised to the $k$-th power.  Because roots of unity are traditionally
defined only for $n>0$, we adopt the convention where \AF{-ωᶠ} $n$ denotes
$(n+1)$-th root of unity.  In other words we (morally) interpret \AF{-ωᶠ} \AB{n}
\AB{k} as $\exp ({-i2πk} / (1 + n))$ thereby avoiding the degenerate case
$n=0$.
\begin{code}[hide]%
\>[0]\AgdaKeyword{module}\AgdaSpace{}%
\AgdaModule{Interface}\AgdaSpace{}%
\AgdaKeyword{where}\<%
\\
\>[0][@{}l@{\AgdaIndent{0}}]%
\>[2]\AgdaKeyword{postulate}\<%
\end{code}
\begin{mathpar}
\codeblock{\begin{code}%
\>[2][@{}l@{\AgdaIndent{1}}]%
\>[4]\AgdaPostulate{𝔽}%
\>[10]\AgdaSymbol{:}\AgdaSpace{}%
\AgdaPrimitive{Set}\<%
\\
\>[4]\AgdaOperator{\AgdaPostulate{\AgdaUnderscore{}+ᶠ\AgdaUnderscore{}}}%
\>[10]\AgdaSymbol{:}\AgdaSpace{}%
\AgdaPostulate{𝔽}\AgdaSpace{}%
\AgdaSymbol{→}\AgdaSpace{}%
\AgdaPostulate{𝔽}\AgdaSpace{}%
\AgdaSymbol{→}\AgdaSpace{}%
\AgdaPostulate{𝔽}\<%
\\
\>[4]\AgdaOperator{\AgdaPostulate{-ᶠ\AgdaUnderscore{}}}%
\>[10]\AgdaSymbol{:}\AgdaSpace{}%
\AgdaPostulate{𝔽}\AgdaSpace{}%
\AgdaSymbol{→}\AgdaSpace{}%
\AgdaPostulate{𝔽}\<%
\\
\>[4]\AgdaOperator{\AgdaPostulate{\AgdaUnderscore{}*ᶠ\AgdaUnderscore{}}}%
\>[10]\AgdaSymbol{:}\AgdaSpace{}%
\AgdaPostulate{𝔽}\AgdaSpace{}%
\AgdaSymbol{→}\AgdaSpace{}%
\AgdaPostulate{𝔽}\AgdaSpace{}%
\AgdaSymbol{→}\AgdaSpace{}%
\AgdaPostulate{𝔽}\<%
\\
\>[4]\AgdaPostulate{𝟘ᶠ}%
\>[10]\AgdaSymbol{:}\AgdaSpace{}%
\AgdaPostulate{𝔽}\<%
\\
\>[4]\AgdaPostulate{𝟙ᶠ}%
\>[10]\AgdaSymbol{:}\AgdaSpace{}%
\AgdaPostulate{𝔽}\<%
\\
\>[4]\AgdaPostulate{-ωᶠ}%
\>[10]\AgdaSymbol{:}\AgdaSpace{}%
\AgdaSymbol{(}\AgdaBound{n}\AgdaSpace{}%
\AgdaSymbol{:}\AgdaSpace{}%
\AgdaDatatype{ℕ}\AgdaSymbol{)}\AgdaSpace{}%
\AgdaSymbol{→}\AgdaSpace{}%
\AgdaSymbol{(}\AgdaBound{k}\AgdaSpace{}%
\AgdaSymbol{:}\AgdaSpace{}%
\AgdaDatatype{ℕ}\AgdaSymbol{)}\AgdaSpace{}%
\AgdaSymbol{→}\AgdaSpace{}%
\AgdaPostulate{𝔽}\<%
\end{code}
\begin{code}[hide]%
\>[2]\AgdaKeyword{infix}%
\>[9]\AgdaNumber{8}\AgdaSpace{}%
\AgdaOperator{\AgdaPostulate{-ᶠ\AgdaUnderscore{}}}\<%
\\
\>[2]\AgdaKeyword{infixl}\AgdaSpace{}%
\AgdaNumber{7}\AgdaSpace{}%
\AgdaOperator{\AgdaPostulate{\AgdaUnderscore{}*ᶠ\AgdaUnderscore{}}}\<%
\\
\>[2]\AgdaKeyword{infixl}\AgdaSpace{}%
\AgdaNumber{6}\AgdaSpace{}%
\AgdaOperator{\AgdaPostulate{\AgdaUnderscore{}+ᶠ\AgdaUnderscore{}}}\<%
\\
\>[2]\AgdaKeyword{open}\AgdaSpace{}%
\AgdaModule{AlgebraDefinitions}\AgdaSpace{}%
\AgdaSymbol{\{}\AgdaArgument{A}\AgdaSpace{}%
\AgdaSymbol{=}\AgdaSpace{}%
\AgdaPostulate{𝔽}\AgdaSymbol{\}}\AgdaSpace{}%
\AgdaOperator{\AgdaDatatype{\AgdaUnderscore{}≡\AgdaUnderscore{}}}\<%
\\
\>[2]\AgdaKeyword{variable}\<%
\\
\>[2][@{}l@{\AgdaIndent{0}}]%
\>[4]\AgdaGeneralizable{x}\AgdaSpace{}%
\AgdaGeneralizable{y}\AgdaSpace{}%
\AgdaGeneralizable{z}\AgdaSpace{}%
\AgdaSymbol{:}\AgdaSpace{}%
\AgdaDatatype{ℕ}\<%
\\
\>[2]\AgdaKeyword{infixl}\AgdaSpace{}%
\AgdaNumber{6}\AgdaSpace{}%
\AgdaOperator{\AgdaFunction{\AgdaUnderscore{}+ⁿ\AgdaUnderscore{}}}\<%
\\
\>[2]\AgdaKeyword{infixl}\AgdaSpace{}%
\AgdaNumber{7}\AgdaSpace{}%
\AgdaOperator{\AgdaFunction{\AgdaUnderscore{}*ⁿ\AgdaUnderscore{}}}\<%
\\
\>[2]\AgdaOperator{\AgdaFunction{\AgdaUnderscore{}+ⁿ\AgdaUnderscore{}}}\AgdaSpace{}%
\AgdaSymbol{=}\AgdaSpace{}%
\AgdaOperator{\AgdaPrimitive{ℕ.\AgdaUnderscore{}+\AgdaUnderscore{}}}\<%
\\
\>[2]\AgdaOperator{\AgdaFunction{\AgdaUnderscore{}*ⁿ\AgdaUnderscore{}}}\AgdaSpace{}%
\AgdaSymbol{=}\AgdaSpace{}%
\AgdaOperator{\AgdaPrimitive{ℕ.\AgdaUnderscore{}*\AgdaUnderscore{}}}\<%
\\
\>[2]\AgdaKeyword{postulate}\<%
\end{code}
}
\and
\codeblock{\begin{code}%
\>[2][@{}l@{\AgdaIndent{1}}]%
\>[4]\AgdaPostulate{+ᶠ-assoc}%
\>[16]\AgdaSymbol{:}\AgdaSpace{}%
\AgdaFunction{Associative}\AgdaSpace{}%
\AgdaOperator{\AgdaPostulate{\AgdaUnderscore{}+ᶠ\AgdaUnderscore{}}}\<%
\\
\>[4]\AgdaPostulate{+ᶠ-comm}%
\>[16]\AgdaSymbol{:}\AgdaSpace{}%
\AgdaFunction{Commutative}\AgdaSpace{}%
\AgdaOperator{\AgdaPostulate{\AgdaUnderscore{}+ᶠ\AgdaUnderscore{}}}\<%
\\
\>[4]\AgdaPostulate{+ᶠ-idl}%
\>[16]\AgdaSymbol{:}\AgdaSpace{}%
\AgdaFunction{LeftIdentity}\AgdaSpace{}%
\AgdaPostulate{𝟘ᶠ}\AgdaSpace{}%
\AgdaOperator{\AgdaPostulate{\AgdaUnderscore{}+ᶠ\AgdaUnderscore{}}}\<%
\\
\>[4]\AgdaPostulate{+ᶠ-inv}%
\>[16]\AgdaSymbol{:}\AgdaSpace{}%
\AgdaFunction{LeftInverse}\AgdaSpace{}%
\AgdaPostulate{𝟘ᶠ}\AgdaSpace{}%
\AgdaOperator{\AgdaPostulate{-ᶠ\AgdaUnderscore{}}}\AgdaSpace{}%
\AgdaOperator{\AgdaPostulate{\AgdaUnderscore{}+ᶠ\AgdaUnderscore{}}}\<%
\\
\>[4]\AgdaPostulate{*ᶠ-assoc}%
\>[16]\AgdaSymbol{:}\AgdaSpace{}%
\AgdaFunction{Associative}\AgdaSpace{}%
\AgdaOperator{\AgdaPostulate{\AgdaUnderscore{}*ᶠ\AgdaUnderscore{}}}\<%
\\
\>[4]\AgdaPostulate{*ᶠ-comm}%
\>[16]\AgdaSymbol{:}\AgdaSpace{}%
\AgdaFunction{Commutative}\AgdaSpace{}%
\AgdaOperator{\AgdaPostulate{\AgdaUnderscore{}*ᶠ\AgdaUnderscore{}}}\<%
\\
\>[4]\AgdaPostulate{*ᶠ-idl}%
\>[16]\AgdaSymbol{:}\AgdaSpace{}%
\AgdaFunction{LeftIdentity}\AgdaSpace{}%
\AgdaPostulate{𝟙ᶠ}\AgdaSpace{}%
\AgdaOperator{\AgdaPostulate{\AgdaUnderscore{}*ᶠ\AgdaUnderscore{}}}\<%
\\
\>[4]\AgdaPostulate{*ᶠ-+ᶠ-dist}%
\>[16]\AgdaSymbol{:}\AgdaSpace{}%
\AgdaOperator{\AgdaPostulate{\AgdaUnderscore{}*ᶠ\AgdaUnderscore{}}}\AgdaSpace{}%
\AgdaOperator{\AgdaFunction{DistributesOverˡ}}\AgdaSpace{}%
\AgdaOperator{\AgdaPostulate{\AgdaUnderscore{}+ᶠ\AgdaUnderscore{}}}\<%
\end{code}
}
\end{mathpar}
The right-hand side records the algebraic laws required for \AF{𝔽} to
constitute a commutative ring.  These names used in types are defined
in Agda's standard library and expand to the expected formulations.
For instance \AF{Commutative} \AF{\_+ᶠ\_} expands to: ∀ \AB{x} \AB{y} →
\AB{x} \AF{+ᶠ} \AB{y} \AF{≡} \AB{y} \AF{+ᶠ} \AB{x}.
We now postulate the properties of \AF{-ωᶠ} that stem from the algebra
of exponents.  Interpreting $z = e^{-2\pi}$, the intended equalities are:
$$
  z^{0/x} = 1\qquad
  z^{yx/x} = 1\qquad
  z^{xz/xy} = z^{z/y}\qquad
  z^{(y + z)/x} = z^{y/x}z^{z/x}
$$
Because of our encoding, these statements must be expressed in a slightly more
involved form, which we capture in the following Agda types. These lemmas
constitute the logical backbone of the subsequent development.
\begin{code}%
\>[4]\AgdaPostulate{ω-x-0}%
\>[13]\AgdaSymbol{:}%
\>[16]\AgdaPostulate{-ωᶠ}\AgdaSpace{}%
\AgdaGeneralizable{x}\AgdaSpace{}%
\AgdaNumber{0}\AgdaSpace{}%
\AgdaOperator{\AgdaDatatype{≡}}\AgdaSpace{}%
\AgdaPostulate{𝟙ᶠ}\<%
\\
\>[4]\AgdaPostulate{ω-x-yx}%
\>[13]\AgdaSymbol{:}%
\>[16]\AgdaPostulate{-ωᶠ}\AgdaSpace{}%
\AgdaGeneralizable{x}\AgdaSpace{}%
\AgdaSymbol{((}\AgdaNumber{1}\AgdaSpace{}%
\AgdaOperator{\AgdaFunction{+ⁿ}}\AgdaSpace{}%
\AgdaGeneralizable{x}\AgdaSymbol{)}\AgdaSpace{}%
\AgdaOperator{\AgdaFunction{*ⁿ}}\AgdaSpace{}%
\AgdaGeneralizable{y}\AgdaSymbol{)}\AgdaSpace{}%
\AgdaOperator{\AgdaDatatype{≡}}\AgdaSpace{}%
\AgdaPostulate{𝟙ᶠ}\<%
\\
\>[4]\AgdaPostulate{ω-xy-xz}%
\>[13]\AgdaSymbol{:}%
\>[16]\AgdaPostulate{-ωᶠ}\AgdaSpace{}%
\AgdaSymbol{(}\AgdaGeneralizable{x}\AgdaSpace{}%
\AgdaOperator{\AgdaFunction{+ⁿ}}\AgdaSpace{}%
\AgdaGeneralizable{y}\AgdaSpace{}%
\AgdaOperator{\AgdaFunction{+ⁿ}}\AgdaSpace{}%
\AgdaSymbol{(}\AgdaGeneralizable{x}\AgdaSpace{}%
\AgdaOperator{\AgdaFunction{*ⁿ}}\AgdaSpace{}%
\AgdaGeneralizable{y}\AgdaSymbol{))}\AgdaSpace{}%
\AgdaSymbol{((}\AgdaNumber{1}\AgdaSpace{}%
\AgdaOperator{\AgdaFunction{+ⁿ}}\AgdaSpace{}%
\AgdaGeneralizable{x}\AgdaSymbol{)}\AgdaSpace{}%
\AgdaOperator{\AgdaFunction{*ⁿ}}\AgdaSpace{}%
\AgdaGeneralizable{z}\AgdaSymbol{)}\AgdaSpace{}%
\AgdaOperator{\AgdaDatatype{≡}}\AgdaSpace{}%
\AgdaPostulate{-ωᶠ}\AgdaSpace{}%
\AgdaGeneralizable{y}\AgdaSpace{}%
\AgdaGeneralizable{z}\<%
\\
\>[4]\AgdaPostulate{ω-x-y+z}%
\>[13]\AgdaSymbol{:}%
\>[16]\AgdaPostulate{-ωᶠ}\AgdaSpace{}%
\AgdaGeneralizable{x}\AgdaSpace{}%
\AgdaSymbol{(}\AgdaGeneralizable{y}\AgdaSpace{}%
\AgdaOperator{\AgdaFunction{+ⁿ}}\AgdaSpace{}%
\AgdaGeneralizable{z}\AgdaSymbol{)}\AgdaSpace{}%
\AgdaOperator{\AgdaDatatype{≡}}\AgdaSpace{}%
\AgdaPostulate{-ωᶠ}\AgdaSpace{}%
\AgdaGeneralizable{x}\AgdaSpace{}%
\AgdaGeneralizable{y}\AgdaSpace{}%
\AgdaOperator{\AgdaPostulate{*ᶠ}}\AgdaSpace{}%
\AgdaPostulate{-ωᶠ}\AgdaSpace{}%
\AgdaGeneralizable{x}\AgdaSpace{}%
\AgdaGeneralizable{z}\<%
\end{code}
Next, on the left we introduce an axiomatic theory of the real numbers.  The signature
includes the usual arithmetic operations, a primitive embedding of natural
numbers \AF{ι}, a unary minus, an multiplicative ``inverse'' operation \AF{inv}
which should be thought of as $\AF{inv}\ x = 1/(1 + x)$, and the elementary
trigonometric functions \AF{sin} and \AF{cos}.  In the middle we we define derived
operations: binary subtraction through unary minus, ``division'' through \AF{inv}
and the notions of zero and one through \AF{ι}.  After that we define properties of
injection \AF{ι}. On the right we require reals to be a commutative ring.
\begin{code}[hide]%
\>[0]\AgdaKeyword{postulate}\<%
\end{code}
\begin{mathpar}
\codeblock{\begin{code}%
\>[0][@{}l@{\AgdaIndent{1}}]%
\>[2]\AgdaPostulate{ℝ}%
\>[9]\AgdaSymbol{:}\AgdaSpace{}%
\AgdaPrimitive{Set}\<%
\\
\>[2]\AgdaPostulate{π}%
\>[9]\AgdaSymbol{:}\AgdaSpace{}%
\AgdaPostulate{ℝ}\<%
\\
\>[2]\AgdaOperator{\AgdaPostulate{\AgdaUnderscore{}+\AgdaUnderscore{}}}%
\>[9]\AgdaSymbol{:}\AgdaSpace{}%
\AgdaPostulate{ℝ}\AgdaSpace{}%
\AgdaSymbol{→}\AgdaSpace{}%
\AgdaPostulate{ℝ}\AgdaSpace{}%
\AgdaSymbol{→}\AgdaSpace{}%
\AgdaPostulate{ℝ}\<%
\\
\>[2]\AgdaOperator{\AgdaPostulate{\AgdaUnderscore{}*\AgdaUnderscore{}}}%
\>[9]\AgdaSymbol{:}\AgdaSpace{}%
\AgdaPostulate{ℝ}\AgdaSpace{}%
\AgdaSymbol{→}\AgdaSpace{}%
\AgdaPostulate{ℝ}\AgdaSpace{}%
\AgdaSymbol{→}\AgdaSpace{}%
\AgdaPostulate{ℝ}\<%
\\
\>[2]\AgdaPostulate{ι}%
\>[9]\AgdaSymbol{:}\AgdaSpace{}%
\AgdaDatatype{ℕ}\AgdaSpace{}%
\AgdaSymbol{→}\AgdaSpace{}%
\AgdaPostulate{ℝ}\<%
\\
\>[2]\AgdaOperator{\AgdaPostulate{-\AgdaUnderscore{}}}%
\>[9]\AgdaSymbol{:}\AgdaSpace{}%
\AgdaPostulate{ℝ}\AgdaSpace{}%
\AgdaSymbol{→}\AgdaSpace{}%
\AgdaPostulate{ℝ}\<%
\\
\>[2]\AgdaPostulate{inv}%
\>[9]\AgdaSymbol{:}\AgdaSpace{}%
\AgdaDatatype{ℕ}\AgdaSpace{}%
\AgdaSymbol{→}\AgdaSpace{}%
\AgdaPostulate{ℝ}\<%
\\
\>[2]\AgdaPostulate{sin}%
\>[9]\AgdaSymbol{:}\AgdaSpace{}%
\AgdaPostulate{ℝ}\AgdaSpace{}%
\AgdaSymbol{→}\AgdaSpace{}%
\AgdaPostulate{ℝ}\<%
\\
\>[2]\AgdaPostulate{cos}%
\>[9]\AgdaSymbol{:}\AgdaSpace{}%
\AgdaPostulate{ℝ}\AgdaSpace{}%
\AgdaSymbol{→}\AgdaSpace{}%
\AgdaPostulate{ℝ}\<%
\end{code}
\begin{code}[hide]%
\>[0]\AgdaKeyword{infixl}\AgdaSpace{}%
\AgdaNumber{6}\AgdaSpace{}%
\AgdaOperator{\AgdaFunction{\AgdaUnderscore{}-\AgdaUnderscore{}}}\<%
\\
\>[0]\AgdaKeyword{infixl}\AgdaSpace{}%
\AgdaNumber{6}\AgdaSpace{}%
\AgdaOperator{\AgdaPostulate{\AgdaUnderscore{}+\AgdaUnderscore{}}}\<%
\\
\>[0]\AgdaKeyword{infixl}\AgdaSpace{}%
\AgdaNumber{6}\AgdaSpace{}%
\AgdaOperator{\AgdaPostulate{-\AgdaUnderscore{}}}\<%
\\
\>[0]\AgdaKeyword{infixl}\AgdaSpace{}%
\AgdaNumber{10}\AgdaSpace{}%
\AgdaOperator{\AgdaPostulate{\AgdaUnderscore{}*\AgdaUnderscore{}}}\<%
\\
\>[0]\AgdaKeyword{infixl}\AgdaSpace{}%
\AgdaNumber{10}\AgdaSpace{}%
\AgdaOperator{\AgdaFunction{\AgdaUnderscore{}/\AgdaUnderscore{}}}\<%
\\
\>[0]\AgdaKeyword{open}\AgdaSpace{}%
\AgdaModule{AlgebraDefinitions}\AgdaSpace{}%
\AgdaSymbol{\{}\AgdaArgument{A}\AgdaSpace{}%
\AgdaSymbol{=}\AgdaSpace{}%
\AgdaPostulate{ℝ}\AgdaSymbol{\}}\AgdaSpace{}%
\AgdaOperator{\AgdaDatatype{\AgdaUnderscore{}≡\AgdaUnderscore{}}}\<%
\\
\\%
\>[0]\AgdaKeyword{variable}\<%
\\
\>[0][@{}l@{\AgdaIndent{0}}]%
\>[2]\AgdaGeneralizable{x}\AgdaSpace{}%
\AgdaGeneralizable{y}\AgdaSpace{}%
\AgdaGeneralizable{z}\AgdaSpace{}%
\AgdaSymbol{:}\AgdaSpace{}%
\AgdaPostulate{ℝ}\<%
\\
\>[2]\AgdaGeneralizable{m}\AgdaSpace{}%
\AgdaGeneralizable{n}\AgdaSpace{}%
\AgdaGeneralizable{k}\AgdaSpace{}%
\AgdaSymbol{:}\AgdaSpace{}%
\AgdaDatatype{ℕ}\<%
\\
\\%
\>[0]\AgdaKeyword{infixl}\AgdaSpace{}%
\AgdaNumber{6}\AgdaSpace{}%
\AgdaOperator{\AgdaFunction{\AgdaUnderscore{}+ⁿ\AgdaUnderscore{}}}\<%
\\
\>[0]\AgdaKeyword{infixl}\AgdaSpace{}%
\AgdaNumber{7}\AgdaSpace{}%
\AgdaOperator{\AgdaFunction{\AgdaUnderscore{}*ⁿ\AgdaUnderscore{}}}\<%
\\
\>[0]\AgdaOperator{\AgdaFunction{\AgdaUnderscore{}+ⁿ\AgdaUnderscore{}}}\AgdaSpace{}%
\AgdaSymbol{=}\AgdaSpace{}%
\AgdaOperator{\AgdaPrimitive{ℕ.\AgdaUnderscore{}+\AgdaUnderscore{}}}\<%
\\
\>[0]\AgdaOperator{\AgdaFunction{\AgdaUnderscore{}*ⁿ\AgdaUnderscore{}}}\AgdaSpace{}%
\AgdaSymbol{=}\AgdaSpace{}%
\AgdaOperator{\AgdaPrimitive{ℕ.\AgdaUnderscore{}*\AgdaUnderscore{}}}\<%
\\
\>[0]\AgdaOperator{\AgdaFunction{\AgdaUnderscore{}-\AgdaUnderscore{}}}\AgdaSpace{}%
\AgdaSymbol{:}\AgdaSpace{}%
\AgdaPostulate{ℝ}\AgdaSpace{}%
\AgdaSymbol{→}\AgdaSpace{}%
\AgdaPostulate{ℝ}\AgdaSpace{}%
\AgdaSymbol{→}\AgdaSpace{}%
\AgdaPostulate{ℝ}\<%
\\
\>[0]\AgdaOperator{\AgdaFunction{\AgdaUnderscore{}/\AgdaUnderscore{}}}\AgdaSpace{}%
\AgdaSymbol{:}\AgdaSpace{}%
\AgdaPostulate{ℝ}\AgdaSpace{}%
\AgdaSymbol{→}\AgdaSpace{}%
\AgdaDatatype{ℕ}\AgdaSpace{}%
\AgdaSymbol{→}\AgdaSpace{}%
\AgdaPostulate{ℝ}\<%
\end{code}
}
\codeblock{
\begin{AgdaSuppressSpace}
\begin{code}%
\>[0]\AgdaBound{a}\AgdaSpace{}%
\AgdaOperator{\AgdaFunction{-}}\AgdaSpace{}%
\AgdaBound{b}\AgdaSpace{}%
\AgdaSymbol{=}\AgdaSpace{}%
\AgdaBound{a}\AgdaSpace{}%
\AgdaOperator{\AgdaPostulate{+}}\AgdaSpace{}%
\AgdaSymbol{(}\AgdaOperator{\AgdaPostulate{-}}\AgdaSpace{}%
\AgdaBound{b}\AgdaSymbol{)}\<%
\\
\>[0]\AgdaFunction{𝟘}\AgdaSpace{}%
\AgdaSymbol{=}\AgdaSpace{}%
\AgdaPostulate{ι}\AgdaSpace{}%
\AgdaNumber{0}\<%
\\
\>[0]\AgdaFunction{𝟙}\AgdaSpace{}%
\AgdaSymbol{=}\AgdaSpace{}%
\AgdaPostulate{ι}\AgdaSpace{}%
\AgdaNumber{1}\<%
\\
\>[0]\AgdaBound{a}\AgdaSpace{}%
\AgdaOperator{\AgdaFunction{/}}\AgdaSpace{}%
\AgdaBound{n}\AgdaSpace{}%
\AgdaSymbol{=}\AgdaSpace{}%
\AgdaBound{a}\AgdaSpace{}%
\AgdaOperator{\AgdaPostulate{*}}\AgdaSpace{}%
\AgdaPostulate{inv}\AgdaSpace{}%
\AgdaBound{n}\<%
\end{code}
\begin{code}[hide]%
\>[0]\AgdaKeyword{postulate}\<%
\end{code}
\begin{code}%
\>[0][@{}l@{\AgdaIndent{1}}]%
\>[2]\AgdaPostulate{inv-*}\AgdaSpace{}%
\AgdaSymbol{:}\AgdaSpace{}%
\AgdaPostulate{ι}\AgdaSpace{}%
\AgdaSymbol{(}\AgdaNumber{1}\AgdaSpace{}%
\AgdaOperator{\AgdaFunction{+ⁿ}}\AgdaSpace{}%
\AgdaGeneralizable{n}\AgdaSymbol{)}\AgdaSpace{}%
\AgdaOperator{\AgdaPostulate{*}}\AgdaSpace{}%
\AgdaPostulate{inv}\AgdaSpace{}%
\AgdaGeneralizable{n}\AgdaSpace{}%
\AgdaOperator{\AgdaDatatype{≡}}\AgdaSpace{}%
\AgdaFunction{𝟙}\<%
\\
\>[2]\AgdaPostulate{ι-inj}\AgdaSpace{}%
\AgdaSymbol{:}\AgdaSpace{}%
\AgdaPostulate{ι}\AgdaSpace{}%
\AgdaGeneralizable{n}\AgdaSpace{}%
\AgdaOperator{\AgdaDatatype{≡}}\AgdaSpace{}%
\AgdaPostulate{ι}\AgdaSpace{}%
\AgdaGeneralizable{m}\AgdaSpace{}%
\AgdaSymbol{→}\AgdaSpace{}%
\AgdaGeneralizable{m}\AgdaSpace{}%
\AgdaOperator{\AgdaDatatype{≡}}\AgdaSpace{}%
\AgdaGeneralizable{n}\<%
\\
\>[2]\AgdaPostulate{ι-+}\AgdaSpace{}%
\AgdaSymbol{:}\AgdaSpace{}%
\AgdaPostulate{ι}\AgdaSpace{}%
\AgdaSymbol{(}\AgdaGeneralizable{m}\AgdaSpace{}%
\AgdaOperator{\AgdaFunction{+ⁿ}}\AgdaSpace{}%
\AgdaGeneralizable{n}\AgdaSymbol{)}\AgdaSpace{}%
\AgdaOperator{\AgdaDatatype{≡}}\AgdaSpace{}%
\AgdaPostulate{ι}\AgdaSpace{}%
\AgdaGeneralizable{m}\AgdaSpace{}%
\AgdaOperator{\AgdaPostulate{+}}\AgdaSpace{}%
\AgdaPostulate{ι}\AgdaSpace{}%
\AgdaGeneralizable{n}\<%
\\
\>[2]\AgdaPostulate{ι-*}\AgdaSpace{}%
\AgdaSymbol{:}\AgdaSpace{}%
\AgdaPostulate{ι}\AgdaSpace{}%
\AgdaSymbol{(}\AgdaGeneralizable{m}\AgdaSpace{}%
\AgdaOperator{\AgdaFunction{*ⁿ}}\AgdaSpace{}%
\AgdaGeneralizable{n}\AgdaSymbol{)}\AgdaSpace{}%
\AgdaOperator{\AgdaDatatype{≡}}\AgdaSpace{}%
\AgdaPostulate{ι}\AgdaSpace{}%
\AgdaGeneralizable{m}\AgdaSpace{}%
\AgdaOperator{\AgdaPostulate{*}}\AgdaSpace{}%
\AgdaPostulate{ι}\AgdaSpace{}%
\AgdaGeneralizable{n}\<%
\end{code}
\end{AgdaSuppressSpace}
}
\codeblock{\begin{code}%
\>[2]\AgdaPostulate{+-assoc}\AgdaSpace{}%
\AgdaSymbol{:}\AgdaSpace{}%
\AgdaFunction{Associative}\AgdaSpace{}%
\AgdaOperator{\AgdaPostulate{\AgdaUnderscore{}+\AgdaUnderscore{}}}\<%
\\
\>[2]\AgdaPostulate{+-comm}\AgdaSpace{}%
\AgdaSymbol{:}\AgdaSpace{}%
\AgdaFunction{Commutative}\AgdaSpace{}%
\AgdaOperator{\AgdaPostulate{\AgdaUnderscore{}+\AgdaUnderscore{}}}\<%
\\
\>[2]\AgdaPostulate{+-idl}\AgdaSpace{}%
\AgdaSymbol{:}\AgdaSpace{}%
\AgdaFunction{LeftIdentity}\AgdaSpace{}%
\AgdaFunction{𝟘}\AgdaSpace{}%
\AgdaOperator{\AgdaPostulate{\AgdaUnderscore{}+\AgdaUnderscore{}}}\<%
\\
\>[2]\AgdaPostulate{+-invl}\AgdaSpace{}%
\AgdaSymbol{:}\AgdaSpace{}%
\AgdaFunction{LeftInverse}\AgdaSpace{}%
\AgdaFunction{𝟘}\AgdaSpace{}%
\AgdaOperator{\AgdaPostulate{-\AgdaUnderscore{}}}\AgdaSpace{}%
\AgdaOperator{\AgdaPostulate{\AgdaUnderscore{}+\AgdaUnderscore{}}}\<%
\\
\>[2]\AgdaPostulate{*-assoc}\AgdaSpace{}%
\AgdaSymbol{:}\AgdaSpace{}%
\AgdaFunction{Associative}\AgdaSpace{}%
\AgdaOperator{\AgdaPostulate{\AgdaUnderscore{}*\AgdaUnderscore{}}}\<%
\\
\>[2]\AgdaPostulate{*-comm}\AgdaSpace{}%
\AgdaSymbol{:}\AgdaSpace{}%
\AgdaFunction{Commutative}\AgdaSpace{}%
\AgdaOperator{\AgdaPostulate{\AgdaUnderscore{}*\AgdaUnderscore{}}}\<%
\\
\>[2]\AgdaPostulate{*-idl}\AgdaSpace{}%
\AgdaSymbol{:}\AgdaSpace{}%
\AgdaFunction{LeftIdentity}\AgdaSpace{}%
\AgdaFunction{𝟙}\AgdaSpace{}%
\AgdaOperator{\AgdaPostulate{\AgdaUnderscore{}*\AgdaUnderscore{}}}\<%
\\
\>[2]\AgdaPostulate{*-+-distl}\AgdaSpace{}%
\AgdaSymbol{:}\AgdaSpace{}%
\AgdaOperator{\AgdaPostulate{\AgdaUnderscore{}*\AgdaUnderscore{}}}\AgdaSpace{}%
\AgdaOperator{\AgdaFunction{DistributesOverˡ}}\AgdaSpace{}%
\AgdaOperator{\AgdaPostulate{\AgdaUnderscore{}+\AgdaUnderscore{}}}\<%
\end{code}
}
\end{mathpar}
The trigonometric equalities that follow are sufficient to derive the
usual identities for complex exponentiation:
\begin{mathpar}
\codeblock{\begin{code}%
\>[2]\AgdaPostulate{sin[-x]}%
\>[11]\AgdaSymbol{:}\AgdaSpace{}%
\AgdaPostulate{sin}\AgdaSpace{}%
\AgdaSymbol{(}\AgdaOperator{\AgdaPostulate{-}}\AgdaSpace{}%
\AgdaGeneralizable{x}\AgdaSymbol{)}\AgdaSpace{}%
\AgdaOperator{\AgdaDatatype{≡}}\AgdaSpace{}%
\AgdaOperator{\AgdaPostulate{-}}\AgdaSpace{}%
\AgdaPostulate{sin}\AgdaSpace{}%
\AgdaGeneralizable{x}\<%
\\
\>[2]\AgdaPostulate{cos[-x]}%
\>[11]\AgdaSymbol{:}\AgdaSpace{}%
\AgdaPostulate{cos}\AgdaSpace{}%
\AgdaSymbol{(}\AgdaOperator{\AgdaPostulate{-}}\AgdaSpace{}%
\AgdaGeneralizable{x}\AgdaSymbol{)}\AgdaSpace{}%
\AgdaOperator{\AgdaDatatype{≡}}%
\>[27]\AgdaPostulate{cos}\AgdaSpace{}%
\AgdaGeneralizable{x}\<%
\\
\>[2]\AgdaPostulate{sin[2πn]}\AgdaSpace{}%
\AgdaSymbol{:}\AgdaSpace{}%
\AgdaPostulate{sin}\AgdaSpace{}%
\AgdaSymbol{(}\AgdaPostulate{ι}\AgdaSpace{}%
\AgdaNumber{2}\AgdaSpace{}%
\AgdaOperator{\AgdaPostulate{*}}\AgdaSpace{}%
\AgdaPostulate{π}\AgdaSpace{}%
\AgdaOperator{\AgdaPostulate{*}}\AgdaSpace{}%
\AgdaPostulate{ι}\AgdaSpace{}%
\AgdaGeneralizable{n}\AgdaSymbol{)}\AgdaSpace{}%
\AgdaOperator{\AgdaDatatype{≡}}\AgdaSpace{}%
\AgdaFunction{𝟘}\<%
\end{code}
}
\codeblock{\begin{code}%
\>[2]\AgdaPostulate{cos[2πn]}\AgdaSpace{}%
\AgdaSymbol{:}\AgdaSpace{}%
\AgdaPostulate{cos}\AgdaSpace{}%
\AgdaSymbol{(}\AgdaPostulate{ι}\AgdaSpace{}%
\AgdaNumber{2}\AgdaSpace{}%
\AgdaOperator{\AgdaPostulate{*}}\AgdaSpace{}%
\AgdaPostulate{π}\AgdaSpace{}%
\AgdaOperator{\AgdaPostulate{*}}\AgdaSpace{}%
\AgdaPostulate{ι}\AgdaSpace{}%
\AgdaGeneralizable{n}\AgdaSymbol{)}\AgdaSpace{}%
\AgdaOperator{\AgdaDatatype{≡}}\AgdaSpace{}%
\AgdaFunction{𝟙}\<%
\\
\>[2]\AgdaPostulate{sin[x+y]}\AgdaSpace{}%
\AgdaSymbol{:}\AgdaSpace{}%
\AgdaPostulate{sin}\AgdaSpace{}%
\AgdaSymbol{(}\AgdaGeneralizable{x}\AgdaSpace{}%
\AgdaOperator{\AgdaPostulate{+}}\AgdaSpace{}%
\AgdaGeneralizable{y}\AgdaSymbol{)}\AgdaSpace{}%
\AgdaOperator{\AgdaDatatype{≡}}\AgdaSpace{}%
\AgdaPostulate{sin}\AgdaSpace{}%
\AgdaGeneralizable{x}\AgdaSpace{}%
\AgdaOperator{\AgdaPostulate{*}}\AgdaSpace{}%
\AgdaPostulate{cos}\AgdaSpace{}%
\AgdaGeneralizable{y}\AgdaSpace{}%
\AgdaOperator{\AgdaPostulate{+}}\AgdaSpace{}%
\AgdaPostulate{cos}\AgdaSpace{}%
\AgdaGeneralizable{x}\AgdaSpace{}%
\AgdaOperator{\AgdaPostulate{*}}\AgdaSpace{}%
\AgdaPostulate{sin}\AgdaSpace{}%
\AgdaGeneralizable{y}\<%
\\
\>[2]\AgdaPostulate{cos[x+y]}\AgdaSpace{}%
\AgdaSymbol{:}\AgdaSpace{}%
\AgdaPostulate{cos}\AgdaSpace{}%
\AgdaSymbol{(}\AgdaGeneralizable{x}\AgdaSpace{}%
\AgdaOperator{\AgdaPostulate{+}}\AgdaSpace{}%
\AgdaGeneralizable{y}\AgdaSymbol{)}\AgdaSpace{}%
\AgdaOperator{\AgdaDatatype{≡}}\AgdaSpace{}%
\AgdaPostulate{cos}\AgdaSpace{}%
\AgdaGeneralizable{x}\AgdaSpace{}%
\AgdaOperator{\AgdaPostulate{*}}\AgdaSpace{}%
\AgdaPostulate{cos}\AgdaSpace{}%
\AgdaGeneralizable{y}\AgdaSpace{}%
\AgdaOperator{\AgdaFunction{-}}\AgdaSpace{}%
\AgdaPostulate{sin}\AgdaSpace{}%
\AgdaGeneralizable{x}\AgdaSpace{}%
\AgdaOperator{\AgdaPostulate{*}}\AgdaSpace{}%
\AgdaPostulate{sin}\AgdaSpace{}%
\AgdaGeneralizable{y}\<%
\end{code}
}
\end{mathpar}
Having prepared the real and trigonometric infrastructure, we implement complex
numbers as ordered pairs of reals and define the usual arithmetic operations.
Complex exponentiation is defined via the Euler formula, and roots of unity
are obtained by applying this function to the appropriate angle.  Note that
our ``division'' by $n$ aligns with the desired definition of
\AF{-ω} which expects the first argument to be non-zero.
\begin{mathpar}
\codeblock{\begin{code}%
\>[0]\AgdaKeyword{data}\AgdaSpace{}%
\AgdaDatatype{ℂ}\AgdaSpace{}%
\AgdaSymbol{:}\AgdaSpace{}%
\AgdaPrimitive{Set}\AgdaSpace{}%
\AgdaKeyword{where}\<%
\\
\>[0][@{}l@{\AgdaIndent{0}}]%
\>[2]\AgdaOperator{\AgdaInductiveConstructor{\AgdaUnderscore{}+\AgdaUnderscore{}i}}\AgdaSpace{}%
\AgdaSymbol{:}\AgdaSpace{}%
\AgdaSymbol{(}\AgdaBound{re}\AgdaSpace{}%
\AgdaBound{im}\AgdaSpace{}%
\AgdaSymbol{:}\AgdaSpace{}%
\AgdaPostulate{ℝ}\AgdaSymbol{)}\AgdaSpace{}%
\AgdaSymbol{→}\AgdaSpace{}%
\AgdaDatatype{ℂ}\<%
\end{code}
}
\and
\codeblock{\begin{code}%
\>[0]\AgdaFunction{ιᶜ}\AgdaSpace{}%
\AgdaSymbol{:}\AgdaSpace{}%
\AgdaDatatype{ℕ}\AgdaSpace{}%
\AgdaSymbol{→}\AgdaSpace{}%
\AgdaDatatype{ℂ}\<%
\\
\>[0]\AgdaFunction{ιᶜ}\AgdaSpace{}%
\AgdaBound{n}\AgdaSpace{}%
\AgdaSymbol{=}\AgdaSpace{}%
\AgdaPostulate{ι}\AgdaSpace{}%
\AgdaBound{n}\AgdaSpace{}%
\AgdaOperator{\AgdaInductiveConstructor{+}}\AgdaSpace{}%
\AgdaFunction{𝟘}\AgdaSpace{}%
\AgdaOperator{\AgdaInductiveConstructor{i}}\<%
\end{code}
}
\and
\codeblock{\begin{code}%
\>[0]\AgdaFunction{𝟘ᶜ}\AgdaSpace{}%
\AgdaSymbol{=}\AgdaSpace{}%
\AgdaFunction{ιᶜ}\AgdaSpace{}%
\AgdaNumber{0}\<%
\\
\>[0]\AgdaFunction{𝟙ᶜ}\AgdaSpace{}%
\AgdaSymbol{=}\AgdaSpace{}%
\AgdaFunction{ιᶜ}\AgdaSpace{}%
\AgdaNumber{1}\<%
\end{code}
}
\and
\codeblock{\begin{code}%
\>[0]\AgdaOperator{\AgdaFunction{\AgdaUnderscore{}+ᶜ\AgdaUnderscore{}}}\AgdaSpace{}%
\AgdaSymbol{:}\AgdaSpace{}%
\AgdaDatatype{ℂ}\AgdaSpace{}%
\AgdaSymbol{→}\AgdaSpace{}%
\AgdaDatatype{ℂ}\AgdaSpace{}%
\AgdaSymbol{→}\AgdaSpace{}%
\AgdaDatatype{ℂ}\<%
\\
\>[0]\AgdaSymbol{(}\AgdaBound{a}\AgdaSpace{}%
\AgdaOperator{\AgdaInductiveConstructor{+}}\AgdaSpace{}%
\AgdaBound{b}\AgdaSpace{}%
\AgdaOperator{\AgdaInductiveConstructor{i}}\AgdaSymbol{)}\AgdaSpace{}%
\AgdaOperator{\AgdaFunction{+ᶜ}}\AgdaSpace{}%
\AgdaSymbol{(}\AgdaBound{c}\AgdaSpace{}%
\AgdaOperator{\AgdaInductiveConstructor{+}}\AgdaSpace{}%
\AgdaBound{d}\AgdaSpace{}%
\AgdaOperator{\AgdaInductiveConstructor{i}}\AgdaSymbol{)}\AgdaSpace{}%
\AgdaSymbol{=}\AgdaSpace{}%
\AgdaSymbol{(}\AgdaBound{a}\AgdaSpace{}%
\AgdaOperator{\AgdaPostulate{+}}\AgdaSpace{}%
\AgdaBound{c}\AgdaSymbol{)}\AgdaSpace{}%
\AgdaOperator{\AgdaInductiveConstructor{+}}\AgdaSpace{}%
\AgdaSymbol{(}\AgdaBound{b}\AgdaSpace{}%
\AgdaOperator{\AgdaPostulate{+}}\AgdaSpace{}%
\AgdaBound{d}\AgdaSymbol{)}\AgdaSpace{}%
\AgdaOperator{\AgdaInductiveConstructor{i}}\<%
\end{code}
}
\and
\codeblock{\begin{code}%
\>[0]\AgdaOperator{\AgdaFunction{-ᶜ\AgdaUnderscore{}}}\AgdaSpace{}%
\AgdaSymbol{:}\AgdaSpace{}%
\AgdaDatatype{ℂ}\AgdaSpace{}%
\AgdaSymbol{→}\AgdaSpace{}%
\AgdaDatatype{ℂ}\<%
\\
\>[0]\AgdaOperator{\AgdaFunction{-ᶜ}}\AgdaSpace{}%
\AgdaSymbol{(}\AgdaBound{x}\AgdaSpace{}%
\AgdaOperator{\AgdaInductiveConstructor{+}}\AgdaSpace{}%
\AgdaBound{y}\AgdaSpace{}%
\AgdaOperator{\AgdaInductiveConstructor{i}}\AgdaSymbol{)}\AgdaSpace{}%
\AgdaSymbol{=}\AgdaSpace{}%
\AgdaSymbol{(}\AgdaOperator{\AgdaPostulate{-}}\AgdaSpace{}%
\AgdaBound{x}\AgdaSymbol{)}\AgdaSpace{}%
\AgdaOperator{\AgdaInductiveConstructor{+}}\AgdaSpace{}%
\AgdaSymbol{(}\AgdaOperator{\AgdaPostulate{-}}\AgdaSpace{}%
\AgdaBound{y}\AgdaSymbol{)}\AgdaSpace{}%
\AgdaOperator{\AgdaInductiveConstructor{i}}\<%
\end{code}
}
\and
\codeblock{\begin{code}%
\>[0]\AgdaOperator{\AgdaFunction{\AgdaUnderscore{}*ᶜ\AgdaUnderscore{}}}\AgdaSpace{}%
\AgdaSymbol{:}\AgdaSpace{}%
\AgdaDatatype{ℂ}\AgdaSpace{}%
\AgdaSymbol{→}\AgdaSpace{}%
\AgdaDatatype{ℂ}\AgdaSpace{}%
\AgdaSymbol{→}\AgdaSpace{}%
\AgdaDatatype{ℂ}\<%
\\
\>[0]\AgdaSymbol{(}\AgdaBound{a}\AgdaSpace{}%
\AgdaOperator{\AgdaInductiveConstructor{+}}\AgdaSpace{}%
\AgdaBound{b}\AgdaSpace{}%
\AgdaOperator{\AgdaInductiveConstructor{i}}\AgdaSymbol{)}\AgdaSpace{}%
\AgdaOperator{\AgdaFunction{*ᶜ}}\AgdaSpace{}%
\AgdaSymbol{(}\AgdaBound{c}\AgdaSpace{}%
\AgdaOperator{\AgdaInductiveConstructor{+}}\AgdaSpace{}%
\AgdaBound{d}\AgdaSpace{}%
\AgdaOperator{\AgdaInductiveConstructor{i}}\AgdaSymbol{)}\AgdaSpace{}%
\AgdaSymbol{=}\AgdaSpace{}%
\AgdaSymbol{(}\AgdaBound{a}\AgdaSpace{}%
\AgdaOperator{\AgdaPostulate{*}}\AgdaSpace{}%
\AgdaBound{c}\AgdaSpace{}%
\AgdaOperator{\AgdaFunction{-}}\AgdaSpace{}%
\AgdaBound{b}\AgdaSpace{}%
\AgdaOperator{\AgdaPostulate{*}}\AgdaSpace{}%
\AgdaBound{d}\AgdaSymbol{)}\AgdaSpace{}%
\AgdaOperator{\AgdaInductiveConstructor{+}}\AgdaSpace{}%
\AgdaSymbol{(}\AgdaBound{a}\AgdaSpace{}%
\AgdaOperator{\AgdaPostulate{*}}\AgdaSpace{}%
\AgdaBound{d}\AgdaSpace{}%
\AgdaOperator{\AgdaPostulate{+}}\AgdaSpace{}%
\AgdaBound{b}\AgdaSpace{}%
\AgdaOperator{\AgdaPostulate{*}}\AgdaSpace{}%
\AgdaBound{c}\AgdaSymbol{)}\AgdaSpace{}%
\AgdaOperator{\AgdaInductiveConstructor{i}}\<%
\end{code}
}
\and
\codeblock{\begin{code}%
\>[0]\AgdaOperator{\AgdaFunction{e\textasciicircum{}i\AgdaUnderscore{}}}\AgdaSpace{}%
\AgdaSymbol{:}\AgdaSpace{}%
\AgdaPostulate{ℝ}\AgdaSpace{}%
\AgdaSymbol{→}\AgdaSpace{}%
\AgdaDatatype{ℂ}\<%
\\
\>[0]\AgdaOperator{\AgdaFunction{e\textasciicircum{}i\AgdaUnderscore{}}}\AgdaSpace{}%
\AgdaBound{x}\AgdaSpace{}%
\AgdaSymbol{=}\AgdaSpace{}%
\AgdaSymbol{(}\AgdaPostulate{cos}\AgdaSpace{}%
\AgdaBound{x}\AgdaSymbol{)}\AgdaSpace{}%
\AgdaOperator{\AgdaInductiveConstructor{+}}\AgdaSpace{}%
\AgdaSymbol{(}\AgdaPostulate{sin}\AgdaSpace{}%
\AgdaBound{x}\AgdaSymbol{)}\AgdaSpace{}%
\AgdaOperator{\AgdaInductiveConstructor{i}}\<%
\end{code}
}
\and
\codeblock{\begin{code}%
\>[0]\AgdaFunction{-ωᶜ}\AgdaSpace{}%
\AgdaSymbol{:}\AgdaSpace{}%
\AgdaDatatype{ℕ}\AgdaSpace{}%
\AgdaSymbol{→}\AgdaSpace{}%
\AgdaDatatype{ℕ}\AgdaSpace{}%
\AgdaSymbol{→}\AgdaSpace{}%
\AgdaDatatype{ℂ}\<%
\\
\>[0]\AgdaFunction{-ωᶜ}\AgdaSpace{}%
\AgdaBound{n}\AgdaSpace{}%
\AgdaBound{k}\AgdaSpace{}%
\AgdaSymbol{=}\AgdaSpace{}%
\AgdaOperator{\AgdaFunction{e\textasciicircum{}i}}\AgdaSpace{}%
\AgdaSymbol{(}\AgdaOperator{\AgdaPostulate{-}}\AgdaSpace{}%
\AgdaPostulate{ι}\AgdaSpace{}%
\AgdaNumber{2}\AgdaSpace{}%
\AgdaOperator{\AgdaPostulate{*}}\AgdaSpace{}%
\AgdaPostulate{π}\AgdaSpace{}%
\AgdaOperator{\AgdaPostulate{*}}\AgdaSpace{}%
\AgdaPostulate{ι}\AgdaSpace{}%
\AgdaBound{k}\AgdaSpace{}%
\AgdaOperator{\AgdaFunction{/}}\AgdaSpace{}%
\AgdaBound{n}\AgdaSymbol{)}\<%
\end{code}
}
\end{mathpar}
The aim now is to demonstrate that the twelve theorems---eight concerning ring
properties and four concerning roots‑of‑unity properties---introduced at the
beginning of this section hold for \AF{ℂ} under the operations defined above. We
translate the preceding definitions, together with the properties of the real
and natural numbers and apply \vampire to each goal. Detailed
specifications are provided in the supplementary material located in the
subdirectory \texttt{complex}. Here we simply note that all twelve theorems were proved
automatically in a fraction of a second---an outcome that is both gratifying and
somewhat surprising, given that some of the generated proofs comprise roughly
three hundred lines of code. We regard this result as compelling evidence that
the proposed approach is practical.

\section{Discussion and Future Work}
\label{sec:discussion}
While we think our prototype shows considerable promise, it is not yet a full hammer-style system for Agda.
It is plainly missing a premise selector and some automation to execute a complete round-trip.
Agda's reflection API already
exposes facilities such as \texttt{execTC} for running external commands and
\texttt{checkFromStringTC} for parsing results, so full automation of these
steps is feasible within the existing framework.
Additionally, the following improvements would improve the technology further to obtain a robust, high-quality hammer.

\paragraph*{Full Clausal Logic}
The current prototype insists on Horn clauses, both in the input and in proofs.
However, full clausal logic seems achievable with a doubly-negated representation~\cite{burel}, and \cref{sec:friedman} corroborates this hypothesis.
This would not only allow richer types to be exported from Agda, but also permit more \vampire proofs.
In particular, \vampire already has support for various kinds of inductive reasoning~\cite{vampire-induction}, but the proofs cannot currently be recovered as they leave the Horn fragment, even on Horn inputs.
Note that we cannot throw the baby out with the bathwater: \vampire's frontend will happily perform classical operations like simplifying $\lnot \lnot P$ to $P$, so a clausal input form remains necessary.

\paragraph*{A Larger Slice of Agda}
However, this does not preclude performing a more elaborate translation
from rich Agda types to \vampire clauses.
We intend to broaden the fragment of Agda that can be translated and
reconstructed. In particular, supporting parameterised types and indexed
families is important. It is common practice to translate inductive families
into parameterised types together with propositional equalities for certain
constructors; work on indexed containers~\cite{indexed-cont} and
descriptors~\cite{descriptors} makes this idea precise.
We also wish to support term synthesis for proof‑relevant types~\cite{proof-relevant-horn}.
For example, synthesising an inhabitant of type $\mathbb{N}$ can
be reduced to a Vampire goal \texttt{nat}.
This would then be provable using an \emph{axiom} \texttt{zero}.
CoqHammer~\cite{coqhammer} and Lean‑auto~\cite{leanauto} demonstrate that
more elaborate translations for more of Agda are feasible, but the details need to be reconciled
with the approach taken here.
\todo{Fit Harrop in here somehow?}

\paragraph*{Readable Proofs}
The generated proofs mirror \vampire's proofs, producing one Agda definition for each step of a \vampire proof.
This is ideal for inspecting the result and diagnosing problems, but often very verbose.
We already apply some simplification during generation, but further work on a ``proof prettifier'' would substantially improve usability.
Normalising proofs works, but produces very large terms ---
as expected from what is effectively cut-elimination.
Our aim is a sensible
middle ground that produces proofs resembling those a human would write.

\section{Related Work}
Combining automated and interactive theorem proving is a relatively young idea.
Ahrendt et al~\cite{ahrendt} are prescient in their assessment of the promise and challenges involved.
It is encouraging that of six challenges they identify in their Section 2, two are completely solved,
significant progress has been made on large theories and inductive reasoning,
and only mismatched logic and lack of specialisation to the target domain remain unresolved.
Various approaches were tried around this time, including Tammet and Smith's experiments with Gandalf, ALF, and the Horn fragment~\cite{alf-hammer}, closely related to our own.

Today, Isabelle's \emph{sledgehammer}~\cite{sledgehammer} is the archetypal hammer:
the system selects premises from the prover library, exports goals and premises to a variety of ATPs~\cite{hammer-provers},
retrieves the required premises from the ATP proofs, and finally uses internal automation to recover a certified proof.
Recent work~\cite{smtisabelle} recovers very fine-grained proofs from SMT solvers, and in a curious turn of events it was also shown that external ATPs are not crucial for Sledgehammer~\cite{sledgehammering-without-atps}.
In a similar vein, we could also mention HOL(y)Hammer~\cite{holyhammer}, \textsf{MizAR}~\cite{mizar}, and the recent first hammer for Metamath~\cite{metamath-hammer}.

Moving to the dependently-typed world,
CoqHammer~\cite{coqhammer} adapts the hammer paradigm to Rocq's calculus of inductive
constructions: it translates Rocq goals to first‑order ATPs, employs heuristic
and machine‑learned premise selection, and uses internal reconstruction
procedures (often involving proof search) to obtain constructive
proofs. SMTCoq~\cite{smtcoq} adopts a more conservative strategy by importing
and checking SMT/SAT proofs inside Rocq, thereby providing strong trust guarantees
at the cost of greater restriction on the back ends. Both
projects exemplify the substantial reconciliation machinery required when
external tools assume classical reasoning but the host system is constructive.
Work in Lean~\cite{leanhammer,leanauto} shows active efforts to bring
ATP automation to contemporary proof assistants: these projects combine
translation, premise selection and partial reconstruction, and investigate the
trade‑off between lightweight translation layers and deeper kernel integration.

Agda‑specific work is scarce. The
closest integrates Waldmeister~\cite{struth-agda-atp}, identifying (as we do)
a representable sub‑language and translating problems to Waldmeister
and proofs back. Unlike our reflection‑based
approach, that project encodes the sub‑language explicitly within Agda (giving
it semantics and defining syntactic translations), which provides a precise
specification of the fragment but imposes a substantial manual encoding burden
on users and maintainers.
Agda's built‑in automatic search (historically \emph{Agsy})
attempts to synthesise inhabitants for a goal type and is effective
for simple cases, but it rarely suffices where non‑trivial composition of
premises or sophisticated search is required.

\section{Conclusions}
%

We have demonstrated a lightweight integration of Agda with the automated
theorem prover \vampire that yields substantial proof‑automation for tedious
goals while preserving the trusted‑core guarantees of an interactive theorem
prover. By harnessing Agda's reflection facilities to emit \vampire problems and
reconstructing resulting proofs in Agda through a compact Prolog‑based
engine, the workflow mirrors the ``hammer'' paradigm without demanding invasive
modifications to either system. Crucially, the approach is not intrinsically
tied to Agda: any dependently‑typed proof assistant that offers reflective extraction
(\eg{} Lean or Idris) could replace the Agda front-end whilst reusing very similar
\vampire translation and proof‑conversion machinery. Likewise, the method is not
\vampire-specific; with relatively little implementation effort alternative ATPs
can be used, allowing the framework to benefit from future advances in automated reasoning.

Our prototype tackles a challenging ITP/ATP pairing: Agda
provides virtually no built‑in tactics, and Vampire represents a maximalist ATP
with a rich inference repertoire. Yet the system discharges non‑trivial
lemmas---such as those concerning roots-of-unity properties---immediately,
thereby markedly reducing the overall development effort. These
findings indicate that a carefully-chosen fragment of dependent type theory can
serve as a robust bridge between interactive and automated provers, and invite
further exploration of richer fragments, tighter ATP integration, and broader
applicability across dependently-typed languages.

\paragraph*{Statement on Large Language Models}
None of the results or the content here are produced by LLM. However, we did employ tools to spell‑ and style‑check human‑written text, with careful inspection of the outcomes.

\bibliography{ITP26}
\end{document}